\documentclass[11pt]{article}
\usepackage{fullpage}
\usepackage{appendix}
\usepackage[pdftex]{graphicx}
\usepackage{subfigure}
\usepackage{va}
\usepackage[margin= 1in]{geometry}
\usepackage{stmaryrd}
\let\oldbibliography\thebibliography
\renewcommand{\thebibliography}[1]{%
  \oldbibliography{#1}%
  \setlength{\itemsep}{0pt}%
}

\usepackage{amsfonts, amsmath, amssymb, amsthm, constants,comment}
\usepackage[usenames]{color}
\usepackage{dsfont}
\usepackage{algorithm} %ctan.org\pkg\algorithms
\usepackage{algpseudocode}
\usepackage{mathtools}

\newcommand{\ind}{1{\hskip -2.5 pt}\hbox{I}}

\definecolor{darkred}{RGB}{100,0,0}
\definecolor{darkgreen}{RGB}{0,100,0}
\definecolor{darkblue}{RGB}{0,0,150}

\usepackage{hyperref}
\hypersetup{colorlinks=true, linkcolor=red, citecolor=blue, urlcolor=darkgreen}
\hypersetup{colorlinks=true, linkcolor=red, citecolor=blue, urlcolor=darkgreen}
\newtheorem{theorem}{Theorem}
\newtheorem{definition}{Definition}
\newtheorem{lemma}{Lemma}
\newtheorem{proposition}{Proposition}
\newtheorem{remark}{Remark}

\newtheorem{corollary}{Corollary}

\providecommand{\scalT}[2]{\left\langle{#1},{#2}\right\rangle}

\def\argmax{\operatornamewithlimits{arg\,max}}

\title{\textbf{\large{ Robust Phase Retrieval and Super-Resolution from One Bit Coded Diffraction Patterns. }}}

\author{\normalsize{\textsc{
Youssef Mroueh$^\dagger,\star$ }}\\
\small \em $\dagger$ CBCL, McGovern Institute,\&  CSAIL, MIT, USA.\\
\small \em $\star$  LCSL, Massachussetts Institute of Technology and Istituto Italiano di Tecnologia. \\
\\
{\small \tt  ymroueh@mit.edu}}

\date{\today}
\begin{document}

\maketitle

\begin{abstract}

In this paper we study a  realistic setup for phase retrieval, where the signal of interest is modulated or masked and then for each modulation or mask a diffraction pattern is collected, producing a coded diffraction pattern (CDP) \cite{FMasks}. We are interested in the setup where the resolution of the collected CDP  is limited by the Fraunhofer diffraction limit of the imaging system.
 We   investigate a novel approach based on a geometric quantization scheme of phase-less linear measurements into (one-bit) coded diffraction patterns, and  a  corresponding recovery scheme.
The key  novelty in this  approach consists in  comparing pairs of coded diffractions patterns across frequencies: the one bit measurements obtained rely on the order statistics of the un-quantized measurements rather than their values .
This results in a robust phase recovery, and  unlike currently available methods, allows to efficiently 
perform phase recovery from  measurements affected by  severe (possibly unknown) non linear, rank preserving perturbations, such as distortions. 
Another important feature of this  approach consists in the fact that it  enables also super-resolution and blind-deconvolution, beyond the diffraction limit of a given imaging system.

\end{abstract}
%\newpage
%\tableofcontent

\section{Introduction}
\subsection{The Phase Retrieval Problem and the Diffraction Limit}
The problem of phase retrieval is ubiquitous in many areas of imaging science and engineering, where we are able  to measure only  magnitude of measurements. 
The phase recovery problem can be modeled as the problem of  reconstructing a $n$-dimensional  complex  vector $x_0$ given 
only the magnitude of $m$ phase-less linear measurements. Such  a problem arises for example in X-ray crystallography \cite{Harrison93,Liu08}, diffraction imaging \cite{Bunk07,Rodenburg} or microscopy \cite{Miao08},  where one can only measure the intensities of the incoming waves, and  wishes to recover the lost phase in order to be able to reconstruct the desired object. 
Formally speaking  for a given vector $x_0\in \mathbb{C}^n$ (without loss of generality we assume $n$ to be even), we wish to measure $\scalT{a_k}{x}$, but the only available information is of the form:
\begin{equation}
b_k=\theta(|\scalT{a_k}{x}|^2), k=1\dots m, 
\label{eq:general}
\end{equation}
where $a_k$ is a set of sampling vector in $\mathbb{C}^n$, and $\theta$ models possibly unknown non linear perturbations of the values: distortion and exponential noise for instance.\\

Recovering signals from the modulus of their  Fourier transform is at  the core of  the phase retrieval problem. For instance in coherent X-ray crystallography \cite{Liu08}, speckle imaging in astronomy \cite{speckle}, in microscopy \cite{SR}  or more broadly in Fourier optics it follows 
from the Fraunhofer diffraction principle that the optical field at the detector can be approximated by the Fourier transform of the sensed object.
Since light detectors can measure only intensities of the incoming waves the problem is therefore to recover the discrete signal $x_0\in \mathbb{C}^n$ from measurements of the type:
\begin{equation}
b_{k}=\theta\left(\left| \sum_{j=1}^{n}x_{0}[j]e^{-i{2\pi (j-1)\frac{(k-1)}{n}}}\right|^2\right),\quad k\in \Omega,\quad \Omega \subseteq [1,n] 
\label{eq:fourier}
\end{equation}
where $\Omega$ represents  a set of sampled frequencies, and $\theta$ a possibly unknown non-linearity.
When compared to \eqref{eq:general} we note that in \eqref{eq:fourier} $a_{k}$ correspond to a set of sampled complex sinusoids.    
When $\Omega=[1,n]$, we have the full knowledge of the modulus of the Fourier series decomposition  of the signal of interest on whole ranges of frequencies.
In practice due to the Fraunhofer diffraction limit we are able to measure intensities of the Fourier transform within a range of frequencies below the so called  cut-off frequency $f_c$. Hence the information we have available about $x_0$ is a sample of the lower end of its power spectra in the form of the lowest $2f_c + 1$ modulus of the Fourier series coefficients ($f_c$ is an integer). \\
For instance in microscopy imaging  with coherent illumination the object of interest is diffracted through a lens. The lens is characterized by its Point Spread Function (PSF) $h$, and the cut-off frequency 
 $f_c=\frac{2\pi NA}{\nu}$, where $NA$ is the numerical aperture of the lens and $\nu$ is the wavelength of the illumination light. Let $\hat{h}, \hat{x}_0$ be the Fourier transform of the PSF and the signal $x_0$ respectively, for continuous signals we know that:
 $$\hat{h}(w)=0 \quad \text{ for } |w|>f_c.$$ 
Hence we measure  the modulus of the Fourier transform  of $x_0\star h$ rather than $x_0$, where $\star$ denotes the discrete convolution operation.
If we set $\theta=\theta_{h}$, where 
 $$\theta_{h}\left(\left| \sum_{j=1}^{n}x_{0}[j]e^{-i{2\pi (j-1)\frac{(k-1)}{n}}}\right|^2\right) := \left| \sum_{j=1}^{n}(x_0\star h)[j]e^{-i{2\pi (j-1)\frac{(k-1)}{n}}}\right|^2,$$
equation \eqref{eq:fourier} becomes (with some abuse of notation, in re-indexing $k$ between $-\frac{n}{2}$ and $\frac{n}{2}-1$ ):
\begin{equation}
b_k=|\hat{h}|^2_{k}|\hat{x}_0|^2_{k} \quad k\in \llbracket-f_c, f_c\rrbracket   \text{ and } b_k= 0 \text { elsewhere}.
\label{eq:SR1}
\end{equation}
It follows from equation \eqref{eq:SR1}, that in addition to the missing phase problem we are facing a super-resolution problem since the high frequency content of the signal is also lost due to the physical resolution limit induced by the cut-off frequency $f_c$ (See for e.g. \cite{SR}).\\
This paper answers the following questions: 
\textit{
\begin{enumerate}
 \item Robust recovery: Is it possible to  robustly recover  the missing phase from  the power spectra of a signal $x_0$ that is undergoing severe unknown non linear distortions or a stochastic noise?
 \item Phase Recovery , Super-Resolution, and Blind deconvolution: Is it possible to recover the signal from the lower end of its power spectra? In other words is it possible to super-resolve the signal beyond the diffraction limit of a given optical system even if the PSF of that system was unknown (blind deconvolution)?
\end{enumerate}}
 \noindent \textbf{Notations:}
 $\star$ represents  a convolution, and $\odot$  the Hadamard product (component-wise product).
 For $z\in \mathbb{C}$, $|z|^2$ is  squared  complex modulus of $z$.
For $a,a' \in \mathbb{C}^n$, $\scalT{a}{a'}$ is the complex dot product in $\mathbb{C}^n$. For $a \in \mathbb{C}^n, a^*$ is the complex conjugate of a, $||a||_2$ is the norm $2$ of $a$ and .
Let $A$ a complex hermitian matrix in $\mathbb{C}^n$, $||A||_{F}$ denotes the Frobenius norm of $A$,  $||A||$ demotes the operator norm of $A$, $Tr(A)$ denotes the trace of A. Throughout the paper, we denote by $c,C$ positive absolute constants whose values may
change from instance to instance. 
\subsection{Phase Retrieval: Previous Work }
As mentioned in the introduction  the set of sampling vectors we are interested in, is the set of complex sinusoids.  
Before tackling the Fourier based sampling we turn to the setting pioneered by \cite{Candes} where  the set of sampling vectors is randomized, i.e we consider a set of independent measurements defined  by independent and identically distributed Complex Gaussian sensing vectors, 
\begin{equation}\label{GaussVectors}
a_i \in \mathbb{C}^n, \quad \quad a_i \sim \mathcal{N}(0,\frac{1}{2}I_n)+i \mathcal{N}(0,\frac{1}{2}I_n),  \quad i=1\dots m.
\end{equation} 
The (noiseless) phase recovery problem is defined as follow.
\begin{definition}[Phase-less Sensing and Phase Recovery]
Suppose  phase-less sensing measurements 
\begin{equation}\label{Meas}
b_i=|\scalT{a_i}{x_0}|^2 \in \mathbb{R}_{+}, \quad \quad i=1\dots m,
\end{equation}
are given for $x_0 \in \mathbb{C}^n$,  where $a_i,~ i=1, \dots, m$ are random vectors as in \eqref{GaussVectors}.
%$a_i \sim \mathcal{N}(0,\frac{1}{2}I_n)+i \mathcal{N}(0,\frac{1}{2}I_n)$.
The phase recovery problem is
\begin{equation}
\begin{aligned}
& \underset{}{\text{find}~x},\quad
& \text{subject to}\quad 
|\scalT{a_i}{x}|^2=b_i,\quad i=1\dots m. 
\end{aligned}
\label{eq:np}
\end{equation}
\end{definition}

\noindent The above  problem is non convex and in the following we recall recent approaches to provably and efficiently recover $x_0$
from a finite number of measurements.
%
%\loz{one more sentence here}.
%%Computationally tractable approaches based on a (convex) relaxation 

\noindent \textbf{SDP (Convex) Relaxation and  PhaseLift.} 
The PhaseLift approach   \cite{Candes} stems from the  observation that $|\scalT{a_i}{x}|^2=Tr(a_ia^*_ixx^*),$ so that 
%where $x^*$ is the complex conjugate transposed  of $x$. 
 if we let $X=xx^*$, Problem \ref{eq:np} can be written as, 
\begin{equation}
\begin{aligned}\label{PLift}
& \underset{}{\text{find}~X,}
& \text{subject to}\quad
&Tr(a_ia_i^*X)=b_i, \quad i=1\dots m, 
& \quad X \succeq 0, \quad rank(X)=1.
\end{aligned}
\end{equation}
While the above formulation is still non convex (and in fact combinatorially hard because of the rank constraint), a convex relaxation can be 
obtained noting that Problem \ref{PLift} can be written as   a rank minimization problem over the positive semi-definite  cone,
\begin{equation}
\underset{X}{\text{min}}\quad rank(X),\quad  \text{subject to}\quad Tr(a_ia_i^*X)=b_i, ~~i=1\dots m,  \quad  X \succeq 0,
\end{equation}
and then  considering the   trace as a surrogate for the rank  \cite{Candes},
\begin{equation}
\underset{X}{\text{min}}\quad Tr(X),\quad  \text{subject to}\quad Tr(a_ia_i^*X)=b_i, ~~i=1\dots m,  \quad  X \succeq 0.
\end{equation}
Indeed, the above problem is  convex and can be solved via semidefinite programming (SDP). 
Intestingly, a different  relaxation is obtained  in \cite{Demanet}
by  ignoring the rank constraint in Problem \ref{PLift}. 
The results in \cite{Candes,Demanet}  show that, with high probability,  the solution $\hat{X}_m$ obtained via either one of the above  relaxations can recover $x_0$ {\em exactly} i.e. $\hat{X}_m=x_0x_0^*$, as soon as  $m\geq c n\log n$.  In fact,  the latter requirement can be further improved to $m\geq c n$ \cite{candesn}.  
 While powerful, the convex relaxation approach incur in cumbersome computations, 
and in practice non convex approaches based on greedy alternating minimization (AM)
%such as proposed by   Gerchberg-Saxton, Fienup and Griffin  
\cite{Gerchberg72,Fienup82,Griffin84} 
are often used. The convergence properties of the latter methods  depend heavily  on the initialization and only 
recently \cite{AM}  they have been shown to globally converge (with high probability) if provided with a suitable initialization.
%
%
%and 
%As we recall in the next section, recent results  showed that provided with suitable  initialization 
%alternating minimization approaches can be provably shown to globally converge with high probability.   
 
\noindent\textbf{Phase Retrieval via Suitably Initialized  Alternating Minimization.}  
Let $A$ be the matrix defined by $m$ sensing vectors as in \eqref{GaussVectors}  and $B=Diag(\sqrt{b})$, where $b$ is the vector of measurements as in \eqref{Meas}. Then 
$$
Ax_0=Bu_0
$$
for $u_0=Ph(Ax_0)$ with  $Ph(z)=\left(\frac{z_1}{|z_1|}, \dots \frac{z_m}{|z_m|}\right)$,  
$z \in \mathbb{C}^m$. 
The above equality suggests the following   natural approach  to recover $(x_0,u_0)$,
\begin{equation}
\underset{x,u}{\text{min}}
%\quad MSE=
||Ax-Bu||^2, \quad \text{subject to}
\quad 
% x\in \mathbb{C}^n,u\in \mathbb{C}^n,
 |u_i|=1, \quad i=1\dots m, \\
\end{equation}
The above problem is not convex  because of the constraint on $u$ and  the AM approach 
consists in optimizing  $u$, for a given $x$, and then optimizing $x$ for a given $u$.
It is easy to see that for a given   $x$,  the optimal $u$ is simply
$u=Ph\left(Ax\right),$ and, for a given $u$, the optimal  $x$ is the solution of  a least square problem. 
The key result in \cite{AM} shows that if such an iteration is initialized with the  maximum eigenvector of the matrix  
\begin{equation}\label{Covb}
\hat{C}_m=\frac{1}{m}\sum_{i=1}^mb_i a_ia_i^*
\end{equation}
then the solution of the alternating minimization $x_{t_0}$ globally converges (with high probability) to the true vector $x_0$. Throughout this paper we call this initialization \emph{SubExp} initialization.
Moreover for a given accuracy $\epsilon, 0<\epsilon<1$,   if  
\begin{equation}\label{AM_SampComp}
m\ge c(n(\log^3 n+\log\frac{1}{\epsilon}\log\log\frac{1}{\epsilon} )),
\end{equation} 
then  $||x_{t_0}-e^{i\phi}x_0||_{2}\leq \epsilon$.\\
\textbf{One Bit Phase Retrieval and Greedy Refinements.}
More recently a new approach for phase retrieval was proposed in \cite{1bitPhase} based on a quantization scheme of severely perturbed phase-less linear measurements.
Assume we observe pairs of independent phase-less measurements:
\begin{equation}\label{DistMeas}
(b^1_i,b^2_i)=\left(\theta(|\scalT{a^1_i}{x_0}|^2),\theta(|\scalT{a^2_i}{x_0}|^2)\right), \quad i=1, \dots, m,
\end{equation}
where $(a^1_i,a^2_i)$ are independent sensing vectors as in \eqref{GaussVectors} and  $\theta$ is a {\em possibly unknown}  rank preserving transformation. 
In particular  $\theta$  models a distortion, e.g. $\theta(s)=\tanh(\alpha s)$, $\alpha\in \mathbb R_+$, 
or an additive noise  $\theta(s)=s+\nu$,   where $\nu$ is a   stochastic  noise, such as an exponential noise.
The  recovery problem from severly perturbed intensity values seems hopeless, and indeed  the key in  this approach is a quantization scheme based on comparing  pairs of phase-less measurements. 
More precisely for each pair $b^1_i, b^2_i$ of measurements of the form  \eqref{DistMeas} we define  $$y_i\in \{-1, 1\}
\quad y_i=sign(b^1_i-  b^2_i),\quad i=1\dots m.$$The one bit phase retrieval problem reduces to a maximum eigenvalue problem induced by the matrix
\begin{equation}\label{Covy}
\hat{C}_m=\frac{1}{m}\sum_{i=1}^{m}  y_i(a^{1}_ia^{1,*}_i-a^{2}_ia^{2,*}_i).
\end{equation}
In \cite{1bitPhase}  it is shown that the expectation of $\hat{C}_m$ satisfies 
$\mathbb E  \hat{C}_m=\lambda x_0x_0^*$, where $\lambda$ is a suitable constant
which depends on $\theta$ and plays the role of a signal-to-noise ratio. Morever for a given accuracy $\epsilon, 0<\epsilon<1$,  if 
$O( \frac{ n\log n }{\epsilon^2 \lambda})$ pairs of measurements  are available, then 
 the solution $\hat{x}_m$ of the above  maximum eigenvalue problem satisfies 
$$||\hat{x}_m-x_0e^{i\phi}||^2_{2}\leq \epsilon,$$  
where  $\phi \in [0,2\pi]$ is a global phase. 
Interestingly authors in \cite{1bitPhase} show that provided 
with the one-bit retrieval initialization,  
the solution of the alternating minimization algorithm $x_{t_0}$ globally converges (with high probability) to the true vector $x_0$ :if  
\begin{equation}\label{AM_SampComp1Bit}
m\ge c(n(\log n+\log\frac{1}{\epsilon}\log\log\frac{1}{\epsilon} )),
\end{equation} 
then  $||x_{t_0}-e^{i\phi}x_0||_{2}\leq \epsilon$. Hence quantization plays the role of a preconditioning that enhances the sample complexity of the overall alternating minimization.
\subsection{Coded Diffraction Patterns and PhaseLift}
While the Gaussian measurements setting allow to carry an  interesting theory and gives a glimpse on the efficiency of proposed methods in  more practical setups, it is of great interest to study the Fourier sampling mentioned in the introduction. 
A practical setup consists in   modulating the signal  with multiple structured illuminations for instance, and then measuring multiple diffraction patterns of the modulated signals. 
The modulation step could be replaced by masking the signal of interest with an appropriate mask.
This is indeed  an attractive framework to resolve the ambiguity in the phase retrieval problem. Firstly suggested in \cite{EM},  this technique comes under different names: digital holography \cite{DH}, ptychography  \cite{nature2}, Fourier ptychtographic microscopy \cite{nature}, etc..
These techniques yield to  many successful applications in structured illumination microscopy \cite{SR} and more broadly in many linear and non linear  Fourier optics applications, where both phase retrieval and super-resolution are achieved via masking or the use of multiple structured illumination modulation.
Let $w_{\ell} \in \mathbb{C}^n,\ell=1\dots r$ be the modulating waves (or the masks), we observe the following coded diffraction patterns  \cite{FMasks}  :
\begin{equation}
b_{\ell,k}=\theta\left(\left| \sum_{j=1}^{n}x_0[j] w_{\ell}[j] e^{-i{2\pi (j-1)\frac{(k-1)}{n}}}\right|^2\right),\quad k\in \Omega,\quad \Omega \subseteq [1,n], \quad \ell=1\dots r.
\end{equation} 
In other words, noting  $F$  the Discrete  Fourier Transform (DFT) Matrix, $Diag(w)$ the diagonal matrix with the modulation pattern on its diagonal, and $\Omega=[1,n]$ we have:
\begin{equation}
b_{\ell}=\theta(|FDiag(w_{\ell})x_0|^2) \in \mathbb{R}^n_+ \quad \ell=1\dots r,
\label{eq:model}
\end{equation}
where $\theta$ and the complex modulus act component-wise.
In a recent work, for $\theta(z)=z$, and  a set of admissible modulations authors in \cite{FMasks}  show that an approach similar to Phase-lift  allows the exact recovery of the signal with high probability given that:
$$r\geq c \log^4n ,$$
for a fixed numerical constant $c$.\\
In \cite{alexeev2012phase, bandeira2013phase} authors introduce another approach to  phase retrieval by polarization.
In \cite{bandeira2013phase} authors propose a construction of binary masks that ensures phase recovery by polarization.  It is shown in \cite{bandeira2013phase} that $O(\log(n))$ binary masks are needed to ensure recovery in the noiseless case.\\
Indeed with this subset of papers on phase retrieval we don't give justice to a large body of papers on that issue  for a  succinct review we refer the reader to \cite{FMasks} and references therein.
    
\subsection{This Paper: One Bit Coded Diffraction Patterns}\label{sec:setup}
In this paper we are interested in the setting where $\theta$ is different from the identity. We restrict our analysis to complex Gaussian modulations.
Three  settings are of interest:
\textit{
\begin{enumerate}
\item Noise:
\subitem $\star$ Additive Stochastic noise: We observe noisy coded diffraction patterns,
\begin{equation}
b_{\ell}=|FDiag(w_{\ell})x_0|^2+\nu_{\ell} \in \mathbb{R}^n_+ \quad \ell=1\dots 2r ,
\label{eq:1bitCDPNoise}
\end{equation}
where $\nu_{\ell}$ are independent exponential vectors $Exp(\gamma), $ $(\sigma=\frac{1}{\gamma^2})$. 
\subitem $\star$ Poisson Noise:  We observe noisy coded diffraction patterns contaminated with poisson noise,
\begin{equation}
b_{\ell}=\mathcal{P}_{\eta}\left(|FDiag(w_{\ell})x_0|^2\right) \in \mathbb{R}^n_+ \quad \ell=1\dots 2r ,
\label{eq:1bitCDPoisson}
\end{equation}
where $\mathcal{P}_{\eta}$ is a component-wise poisson noise : For $z,\eta >0,\mathcal{P}_{\eta}(z)\sim Poisson( \frac{z}{\eta})$.
\item Distortion: We observe  distorted coded diffraction patterns:
\begin{equation}
b_{\ell}=\tanh\left(\alpha|FDiag(w_{\ell})x_0|^2\right) \in \mathbb{R}^n_+ \quad \ell=1\dots 2r ,\alpha>0,
\label{eq:1bitCDPDistortion}
\end{equation}
(with some abuse of notations $\tanh$ acts component-wise).
\item Diffraction Limit/Super-Resolution/Blind deconvolution: The modulated signal diffracts through a lens characterized by a PSF $h$ and a cut-off frequency $f_{c}$. Let $H$ be the Toeplitz matrix associated to $h$, we observe:
\begin{equation}
b_{\ell}=|FHDiag(w_{\ell})x_0|^2 \in \mathbb{R}^n_+ ,\quad \ell=1\dots 2r.
\label{eq:1bitCDPSR}
\end{equation}
\end{enumerate} }
\noindent In this paper we take the point of view of \cite{1bitPhase} and define a quantization scheme for the coded diffraction patterns, by comparing pairs of  coded diffraction patterns.
Consider pairs of coded diffraction patterns associated to pairs of independent modulations $(w^1_i,w^2_i)$, where $w^1_i,w^2_i \sim \mathcal{C}\mathcal{N}(0,I_n)$:
\begin{equation}\label{eq:pairsCDP}
(b^1_{i},b^2_{i})= \left(\theta(|FDiag(w^1_i)x_0|^2),\theta(|FDiag(w^2_i)x_0|^2)\right)\in \mathbb{R}^n_+\times \mathbb{R}^n_+ \quad i=1\dots r.
\end{equation}
For each pair $(b^1_{i},b^2_{i})$ of coded diffraction patterns we define a one bit coded diffraction pattern as:
\begin{equation}
y_i\in \{-1,1\}^n ,\quad  y_i=sign(b^1_i-b^2_i), i=1\dots r.
\label{eq:meas}
\end{equation}
Now the One Bit Phase Retrieval problem consists in finding $x_0$ from the knowledge of  one bit coded diffraction patterns $(y_1\dots y_r)$.
Similarly to One bit phase retrieval form Gaussian measurements we show that the phase retrieval problem from one bit coded diffraction patterns reduces to finding the maximum eigen-vector $\hat{x}_r$
of the matrix $\hat{C}_r$:
\begin{equation}
\hat{C}_r=\frac{1}{r}\sum_{i=1}^r\left(Diag(w^1_i) F Diag(y_i)F^* Diag(w^{1,*}_i) -Diag(w^2_i) F Diag(y_i)F^* Diag(w^{2,*}_i)\right),
\label{eq:Matrix}
\end{equation}
Indeed in this paper we show that : $$\mathbb{E}(\hat{C}_r)=\lambda x_0x_0^*,$$ where $\lambda$ depends on $\theta$.
%Morever for a given accuracy $\epsilon\in [0,1]$,  if 
%$$r\geq c \frac{\log^3 n }{\epsilon^2 \lambda},$$  then 
% the solution of the above  maximum eigenvalue problem $\hat{x}_r$ satisfies 
%$$||\hat{x}_r-x_0e^{i\phi}||^2_{2}\leq \epsilon,$$  
%where  $\phi \in [0,2\pi]$ is a global phase. 

\section{Main Results}

In the following we give the only  assumption we make on $\theta$ throughout the paper, and state our main results  for the three setups of interest  discussed in Section \ref{sec:setup}.  \
As mentioned before, we assume that $\theta$ preserves the ranking of the intensities. 
We shall make one assumption on the non linearity $\theta$,
\begin{equation}{\label{eq:lambda}}
\lambda=\mathbb{E}(\scalT{sign(\theta(E_1)-\theta(E_2))}{(E_1-E_2)})>0,
\end{equation}
where $E_1$, $E_2$ are two independently distributed exponential random $n$-dimensional vectors with mean $\frac{1}{n}$. To see why this assumption is natural, 
notice that $|FDiag(w)x_0|^2\sim  (\frac{1}{n}Exp(1))^{\otimes n}$  if $w \sim \mathcal{C}\mathcal{N}(0,I_n)$ and $||x_0||=1$, thus
$$\mathbb{E}(\scalT{y_i}{|FDiag(w^1_i)x_0|^2-|FDiag(w^2_i)x_0|^2})=\mathbb{E}(\scalT{sign(\theta(E_1)-\theta(E_2))}{(E_1-E_2)})=\lambda>0.$$
Then the above assumption simply means that the one bit measurements preserve robustly the ranking of the intensities. 
Let $\hat{x}_r$ be the maximum eigenvector of $\hat{C}_r$ defined in \eqref{eq:Matrix}.
The following Theorem  shows that $\hat{x}_r$ is an $\epsilon-$ estimate of $x_0$.
 \begin{theorem}[Phase Retrieval From One Bit Coded Diffraction Patterns]
For $x_0 \in \mathbb{C}^n, ||x_0||=1$, and $0<\epsilon<1$. Assume $y_1\dots y_r$, follow the model  given in \eqref{eq:meas}. Then we have with a probability at least $1- O(n^{-2}) $, 
$$\text{for } r \geq \frac{c}{\epsilon^2\lambda^2 } \log^3n, \quad ||\hat{x}_r-x_0e^{i\phi}||^2\leq \epsilon$$
where $c$ is a numeric constant, and $\phi \in [0,2\pi]$ is a global phase.
$\lambda$ is given in \eqref{eq:lambda}.
\label{theo:mainFourier}
\end{theorem}
\noindent For the  noiseless model $\theta(z)=z$ and $\lambda=1$. Thus the theorem states that $O(\log^3n)$ pairs of coded diffraction patterns ensures the recovery of the phase. 
For different observation model $\theta$ it suffices to compute the value of $\lambda$ as given in \eqref{eq:lambda}.
\noindent We turn now to the noisy measurements setup \eqref{eq:1bitCDPNoise} and show robustness of phase retrieval from one bit coded diffraction patterns: 
\begin{corollary}[One bit Recovery/ Noise]
For $x_0 \in \mathbb{C}^n, ||x_0||=1$, and $0<\epsilon<1$. Assume $y_1\dots y_r$, follow the model  given in \eqref{eq:meas}, for $\theta(z)=z+\nu, \nu \sim Exp(\gamma)$.
Where $\nu$ is an exponential noise with variance $\sigma=\frac{1}{\gamma^2}$. 
 Then for any $\epsilon, 0<\epsilon<1$, we have with a probability at least $1- O(n^{-2}) $, 
$$\text{for } r \geq \frac{c  }{\epsilon^2}\frac{(1+\sqrt{\sigma})^4}{(1+2\sqrt{\sigma})^2} \log^3n, \quad ||\hat{x}_r-x_0e^{i\phi}||^{2}\leq \epsilon,$$
where $c$ is a numeric constant, and $\phi \in [0,2\pi]$ is a global phase. 
\label{theo:Noise}
\end{corollary}
\noindent In other words, under an exponential noise we have: $$||\hat{x}_r-x_0e^{i\phi}||^2_{2}\leq C\sqrt{\frac{ \log^3n}{r} }\frac{(1+\sqrt{\sigma})^2}{1+2\sqrt{\sigma}}.$$
Beyond robustness to noise, another desirable feature for phase retrieval from phase-less  measurements, is the robustness to distortions of the values of intensities.
Is it possible to retrieve the phase from coded diffraction patterns that are undergoing clipping for instance (as in equation \eqref{eq:1bitCDPDistortion})? 
\begin{corollary}[One bit Recovery/ Distortion]\label{cor:Distortion}
For $x_0 \in \mathbb{C}^n, ||x_0||=1$, and $\epsilon>0$. Assume $y_1\dots y_m$, follow the model  given in \eqref{eq:meas}, for $\theta(z)=\tanh(\alpha z), \alpha >0$.
 Then for any $ \epsilon, 0<\epsilon<1$, we have with a probability at least $1-O(n^{-2})$, 
$$\text{for } r \geq \frac{c }{\epsilon^2} \frac{ \log^3n}{\lambda^2(\alpha)}, \quad ||\hat{x}_r-x_0e^{i\phi}||^{2}\leq \epsilon,$$
where $c$ is a numeric constant, and $\phi \in [0,2\pi]$ is a global phase. 
$\lambda(\alpha)=\mathbb{E}(|E_1-E_2|sign(1-\tanh(\alpha E_1)\tanh(\alpha E_2))),$ is a decreasing function in $\alpha$.
\end{corollary}
\noindent For the last setup where the resolution of the observed diffraction patterns is limited by the Fraunhofer diffraction limit $f_c$ of an optical system as in equation \eqref{eq:1bitCDPSR}.
We show that the recovery is still possible even if the PSF of the optical system was unknown. The number of modulations needed is poly-logarithmic in the dimension and quadratic in the super-resolution factor $SRF=\frac{n}{2f_c+1}$ defined in \cite{carlos}.   

\begin{corollary}[One bit Recovery/Super-Resolution]
For $x_0 \in \mathbb{C}^n, ||x_0||=1$, and $0<\epsilon<1$. Assume $y_1\dots y_m$, follow the model  given in \eqref{eq:meas} for $(b^1_i,b^2_i)$ defined as in \eqref{eq:1bitCDPSR} for a PSF $h$ characterized by the cut-off frequency $f_c$. Then for $\epsilon, 0<\epsilon<1$, we have with we have with a probability at least $1-O(n^{-2}) $, 
$$\text{for } r \geq \frac{c}{\epsilon^2} (SRF)^2 \log^3n, \quad ||\hat{x}_r-x_0e^{i\phi}||^2\leq \epsilon.$$
where $c$ is numeric constant, and $\phi \in [0,2\pi]$ is a global phase. SRF is the super-Resolution factor defined as: $SRF=\frac{n}{2f_c+1}$.
\label{cor:blur}
\end{corollary}
\noindent It follows that: 
$$r  \geq \frac{c}{\epsilon^2}  \log^3n \quad ||\hat{x}_r-x_0e^{i\phi}||^2\leq SRF \epsilon ,$$
this dependency on the Super-Resolution Factor (SRF) is similar to results in \cite{carlos}, where super-resolution is achieved via total variation norm minimization and linear measurements. Note that in \cite{carlos} the phase of the linear measurements, and the PSF $h$ are assumed to be known.\\
It is worth noting that in  \cite{carlos} the super-resolution problem considered, is different from our setting as authors consider a harder problem: super-resolution from  a single image, and a strong prior, namely  a point sources model and a total variation norm minimization. In our case we have access to multiple coded diffraction patterns and this is known as multi-frame super-resolution see for example \cite{SR} and references therein.\\
\noindent The proof of Corollary \ref{cor:blur} is given in Section \ref{Sec:super}. Corollary \ref{cor:blur}  states a surprising fact: one bit coded diffraction patterns allow not only the  super-resolution of the signal but also it leads to a  blind deconvolution  since the only information needed on $h$ is its super-resolution factor $SRF$, its PSF might be completely unknwon.
Intuitively the random modulations push the high frequency content of $x_0$ to the frequency interval  where the Fourier transform of $h$ is non zero. The number of modulations needed is therefore naturally  proportional to the $SRF$ as shown in Corollary \ref{cor:blur}.  Hence the high frequency content of $x_0$ is mapped to the lower end spectrum  by modulation or masking. Phase retrieval from one bit coded diffraction patterns in a way estimates the missing phase, the missing high frequency content and corrects for the blur induced by the unknown PSF $h$. \\
 Moreover this result is still true if the observation model was:
$$(b^1_{i},b^2_{i})= \left(|FH_{i}Diag(w^1_i)x_0|^2,|FH_{i}Diag(w^2_i)x_0|^2\right)\quad i=1\dots r ,$$
\begin{equation}
y_i\in \{-1,1\}^n ,\quad  y_i=sign(b^1_i-b^2_i),\quad  i=1\dots r,
\end{equation}
where $H_i$ are  Toeplitz matrices associated to different unknown stochastic perturbations $h_i$. We assume for simplicity that the Fourier transform of  $h_i$ are non zero in the same frequency domain (the result is still true if this was not the case we don't analyze this case in this paper ). The only requirement is therefore to have the same perturbation on each considered pairs of coded diffraction patterns. For example in microscopy small perturbations will result in a change in the PSF. In astronomy in speckle imaging different $h_i$ model different atmospheric perturbations in a long exposure acquisition.\\
 Surprisingly one bit coded diffraction patterns allow blind deconvolution even in the case of varying PSFs.

\subsection{Discussion and Perspectives}
\subsubsection{Discussion}
We comment in this section on our results and compare them to the current state of the art and put them in the perspective  of future research.
Let $f_j$ be a row of the DFT matrix.\\
For phase-lift, by inspecting the proof  in \cite{FMasks} we note that three factors govern the sample complexity $O(\log^4(n))$, the first  two of them  come from matrix   concentration inequality and the last one is due to the golfing scheme:
\begin{itemize}
\item A bound on the measurements, $|f_jDiag(w)x_0|^2 j=1\dots m$: $|f_jDiag(w)x_0|^2\leq \beta \log(n)$ with high probability.
\item A bound on the the absolute values of the entries of the modulation $|w_i|^2, i=1\dots n$. In \cite{FMasks}, authors define a family of admissible modulation, such that among other conditions:
$|w_i|^2\leq M$, where $M$ is a constant independent to the dimension. It is worth noting that this class of modulations as opposed to a complex Gaussian modulation, saves extra poly-logarthmic terms  in the overall sample complexity of that approach.
\item An extra $\log(n)$ in the sample complexity is needed for the golfing scheme.
\end{itemize} 
In contrast in our case the saving of extra poly-logarithmic terms $(O(\log^3n))$ comes from the nature of one bit coded diffraction patterns. Our one bit measurements are bounded by one, hence they do not contribute to the sample complexity. On the other hand  our modulations are Complex Gaussian. Complex Gaussian modulations have their squared  absolute values bounded with high probability $|w_i|^2\leq\beta \log(n)$  and hence they contribute to the sample complexity.\\
\subsubsection{Perspectives}\label{Sec:persp}
\noindent\textbf{Greedy Refinements.}
Indeed PhaseLift and one bit phase retrieval are not comparable since one achieves exact recovery and the other achieves approximate recovery. 
For an accuracy $\epsilon$ the sample complexity for one bit phase retrieval scales as $\frac{1}{\epsilon^2}$. As for the Gaussian case the Alternating minimization \cite{AM} for coded diffraction patterns initialized with the one bit solution would guarantee a better dependency on $\epsilon$, we leave that direction to a future research.
We conjecture that $O(\log^3 n+\log\frac{1}{\epsilon}\log\log\frac{1}{\epsilon})$ pairs of  coded diffraction patterns ensures $\epsilon$ recovery with the alternating minimization initialized with the one bit solution.
Our experiments on both simulated data and images confirm that (See Section \ref{num}). \\

\noindent \textbf{Non Gaussian Modulations or Masks.}
Another direction would be to investigate admissible modulation of  \cite{FMasks} and one bit measurements as they both enjoy dimensionless boundedness, we leave also that point to a future work.  \\

\noindent \textbf{Blind Deconvolution from Coded Diffraction Patterns.}
In microscopy  the PSF of the lens is often known.  One Bit solution is agnostic to the PSF, hence one bit phase retrieval offers a good initial point to the alternating minimization conditioned on the knowledge of the PSF: phase retrieval with blur correction (See Section \ref{num}). For an unknown PSF an open question remains on how to provably recover both the signal and the PSF, via alternating minimization suitably initialized.

\subsection{Roadmap}
The paper is organized as follows: In Section \ref{sec:q} we introduce the one bit  coded diffraction patterns scheme  and  the corresponding phase recovery procedure. 
In section \ref{Sec:super} we show how super-resolution can be tackled within our framework. We address  algorithms and computational aspects  in Sections \ref{sec:comp} and \ref{num}.
Finally  we give the proofs in Section \ref{app:Fourier}. 
\section{Quantizing Coded Diffraction Patterns}\label{sec:q}
%\begin{figure}[ht]
%\begin{center}
%\noindent\includegraphics[scale=0.25]{./fig/SImic.png}
%\end{center}
%
%%\section{Coded Diffraction Patterns}
%
%
%\section{Robust One bit Phase Retrieval from Fourier Masks}
%
%While the Gaussian measurements framework captures the main ideas of one bit phase retrieval, a more practical setup is the one that involves Fourier masks.
%In this section we review the phase retrieval problem from Fourier masks, and introduce a quantization scheme and a corresponding recovery problem.
%\subsection{Phase retrieval from Fourier Masks} 
%\subsubsection{Fourier Masks and Structured Illuminations}
%\caption{Acquisition of diffraction patterns of modulated signals with structured illumination via a microscope lens.}
%%\caption{ In blue the $SNR=-\log(||\hat{X}_m-x_0x_0^*||^2_{F})$ the SNR becomes positive and increases as soon as $m \geq n\log(n)$.(Here n=20). }
%\label{fig:obpr}
%\end{figure}
%\subsection{One bit Fourier Masks }
\subsection{Preliminary Matrix Notation}
Let $M \in \mathbb{C}^{n\times n}$ be a complex matrix , $diag(M)$ is a vector in $\mathbb{C}^n$, containing the diagonal elements of $M$.\\
Let $u\in \mathbb{C}^n$ be a complex vector, $Diag(u)$ is a matrix in $\mathbb{C}^{n\times n}$, with $u$ on the diagonal and zeros elsewhere.\\
Let $F$ be the discrete Fourier matrix, such that $F_{jk}=\frac{1}{\sqrt{n}}{e^{-i\frac{2\pi (j-1)(k-1)}{n}}}, j=1\dots n, k=1\dots n$.
%$$F=\frac{1}{\sqrt{n}}\begin{bmatrix}\Omega^{-1\times 1}&\Omega^{-1\times 2}&\ldots&\Omega^{-1\times n}\\\Omega^{-2\times 1}&\Omega^{-2\times 2}&\ldots&\Omega^{-2\times n}\\\vdots&\vdots&\ddots&\vdots\\\Omega^{-n\times 1}&\Omega^{-n\times 2}&\ldots&\Omega^{-n\times n}\end{bmatrix},$$
%where $\Omega=e^{i\frac{2\pi}{n}}$.
%\begin{lemma}\label{lem:manip}
%Let $u \in \mathbb{C}^n$, and $M\in \mathbb{C}^{n\times n}$, we have the following equality: 
%\begin{equation}
%\scalT{y}{diag(M)}=Tr(Diag(y)M).
%\end{equation}
%\end{lemma}
\subsection{One Bit Coded Diffraction Patterns}
We start by defining the quantization scheme of the values of masked Fourier intensities or CDP.
We assume that we observe  $\theta\left(| F Diag(w^1)x_0|^2\right)$,
where $\theta$ is eventually an unknown non linearity satisfying \eqref{eq:lambda}.
Following the same procedure in the Gaussian case we quantize the differential of two independent coded diffraction patterns.

\begin{definition}[One-bit Fourier quantizer]  \label{def:quantizer} 
Let $W=(w^{1},w^{2})$, where $w^{1},w^{2}$ are \text{i.i.d.} 
complex  Gaussian vectors 
%in $\mathbb{C}^n$,\loz{redundant} \textcolor{red}{normalization are important we should specify the half in the variance} 
$\mathcal{N}(0,\frac{1}{2}I_{n})+ i\mathcal{N}(0,\frac{1}{2}I_{n})$.
$w^{1}$ and $w^2$ are called Gaussian masks or modulations.
 Let $F$ be the discrete Fourier matrix in $\mathbb{C}^{n \times n}$.
For $x_0 \in \mathbb{C}^n$, a one bit quantizer of  coded diffraction patterns is given by 
$$Q^{\theta}_{W}:\mathbb{C}^n\to \{-1,1\}^n,\quad Q^{\theta}_{A}(x_0)=sign\left(\theta(| F Diag(w^1)x_0|^2)-\theta(|F Diag(w^2) x_0|^2)\right).$$
where $| F Diag(w^1)x_0|^2$ is the complex modulus of each component of $F Diag(w^1)x_0$. $\theta$ is the observation model.
$\theta$ is eventually an unknown non linearity that satisfies equation  \eqref{eq:lambda}. 
\end{definition}
\noindent Recall that a  basic quantizer in the noiseless case is obtained setting $\theta(z)=z$.\\
Now for a total of $2r$ masks or modulations  we define the one bit  coded diffraction patterns:
\begin{definition}[One Bit Coded Diffraction Patterns] Let $\{W_i=(w^1_i,w^2_i)\}_{1\leq i\leq r}$, be $2r$ i.i.d. Gaussian masks in $\mathbb{C}^n$, and $Q^{\theta}_{W_i}(x_0)$ as in Def \ref{def:quantizer} . The Quantized Phase-less sensing is : $\mathcal{Q}:\mathbb{C}^{n}\to \{-1,1\}^{nr}$ , $\mathcal{Q}(x_0)=(Q^{\theta}_{W_1}(x_0),\dots,Q^{\theta}_{W_m}(x_0))$.
\label{def:sensing}
\end{definition}

\noindent Let  
\begin{equation}\label{eq:modelFourier}
y_i =sign\left(\theta(|F Diag(w^1_i)x|^2)-\theta(|F Diag(w^2_i)x|^2)\right)\in \{-1,1\}^n \quad  i=1\dots r
\end{equation}
In this paper, we are interested in recovering $x_0$ from its  one bit coded diffraction patterns $y= (y_1\dots y_m)=\mathcal{Q}(x_0)$.
It is easy to see that the phase retrieval  amounts to  the following feasibility problem:\\
\begin{equation}
\begin{aligned}
& \underset{}{\text{find}~~ x}\\
& \text{subject to}\\
&\scalT{y_i}{|F Diag(w^1_i)x|^2-|F Diag(w^2_i)x|^2}\geq0,\quad i=1\dots r. \\
&||x||^2=1.\\
\end{aligned}
\label{eq:biPhaseFourier}
\end{equation}
Again we propose the following relaxation to tackle that problem:
\begin{equation}
\begin{aligned}
& \underset{}{\text{$\max_{x,||x||_2=1}$}~~  \left(\frac{1}{r}\sum_{i=1}^{r}\scalT{y_i}{|F Diag(w^1_i)x|^2-|F Diag(w^2_i)x|^2}  \right)}
\end{aligned}
\label{eq:MaxPhase2}
\end{equation}
The proof architecture is similar to the Gaussian case. Proofs are given in Section \ref{app:Fourier}. We start by a preliminary definition:
\begin{definition} [Fourier Risk and Empirical risk]
Let $x_0 \in \mathbb{C}^n, ||x_0||=1$.
For $x \in \mathbb{C}^{n}$ such that $||x||=1$, and $W=\{w^1,w^2\}$ i.i.d. complex Gaussians, let 
\begin{equation*}
\mathcal{E}^{x_0}(x)=\mathbb{E}(\scalT{y}{|F Diag(w^1)x|^2-|F Diag(w^2)x|^2}),
 \end{equation*}
 where $y=sign\left(|F Diag(w^1)x_0|^2-|F Diag(w^2)x_0|^2\right) \in \{-1,1\}^n$.
 Moreover, let
\begin{equation*}
\hat{\mathcal{E}}^{x_0}(x)=\frac{1}{r} \sum_{i=1}^r \scalT{y_i}{|F Diag(w^1_i)x|^2-|F Diag(w^2_i)x|^2}),
\end{equation*}
$y_i =Q^{\theta}_{W_i}(x_0)$ and $W_i=\{(w^1_i,w_i^2)\},i=1\dots r$ are i.i.d. complex Gaussians.\\
 \end{definition} 
\noindent In the following definition the phase retrieval problem is  cast as an empirical risk maximization:
\begin{definition}[Phase retrieval Problem]
The phase retrieval problem amounts to solving:
$$\max_{x,||x||=1} \hat{\mathcal{E}}^{x_0}(x)$$
Let $\hat{x}_r=\arg\max_{x,||x||=1} \hat{\mathcal{E}}^{x_0}(x).$
\end{definition}
\noindent The following proposition shows that the objective function can be written explicitly as a quadratic form.
\begin{proposition}\label{pro:QForm}
$\mathcal{E}^{x_0}(x)$ can be rewritten as the following quadratic form:
\begin{equation}
\mathcal{E}^{x_0}(x)=x^* C x, 
\label{eq:expected}
\end{equation}
where $ C=\mathbb{E}\left(Diag(w^1) F Diag(y)F^* Diag(w^{1,*}) -Diag(w^2) F Diag(y)F^* Diag(w^{2,*})\right).$
and 
\begin{equation}
\hat{\mathcal{E}}^{x_0}(x)=x^*\hat{C}_r x,
\end{equation}
where $\hat{C}_r=\frac{1}{r}\sum_{i=1}^r\left(Diag(w^1_i) F Diag(y_i)F^* Diag(w^{1,*}_i) -Diag(w^2_i) F Diag(y_i)F^* Diag(w^{2,*}_i)\right).$\\
\end{proposition}

\noindent The phase retrieval problem from One bit CDP is therefore a maximum eigenvalue problem, that we call \emph{1bitPhase}:

\begin{equation}\label{eq:Fourier1bit} %\tag{Fourier1bit}
\max_{x,||x||=1}x^*\hat{C}_r x\\
\end{equation}
\newline

%\subsection{Main result for Fourier Masks}
%The following theorem describes the recovery guarantees for the solution   $\hat{x}_r$  of problem Fourier$1$bit  \eqref{eq:Fourier1bit}.
%
%\begin{theorem}[One bit Recovery From Quantized Fourier Gaussian Masks]
%For $x_0 \in \mathbb{C}^n, ||x_0||=1$, and $\epsilon>0$. Assume $y_1\dots y_m$, follows the model  given in \eqref{eq:modelFourier}. Then for any $u>0$, we have with we have with a probability at least $1- C_1\exp(-c_1u^2) $, 
%$$\text{for } r \geq \frac{cu}{\epsilon^2} (\log(n)(1+(\log(n))^2)), \quad ||\hat{x}_r-x_0e^{i\phi}||^2\leq \frac{\epsilon}{\lambda}.$$
%where $c,c_1,C_1$ are universal constants, and $\phi \in [0,2\pi]$ is a global phase shift. 
%$\lambda$ is given in \eqref{eq:lambda}, and has the same values as in the Gaussian case.
%\label{theo:mainFourier}
%\end{theorem}%$\sum_{i=1}^n\mathbb{E}\left( y_i \left(|\scalT{x_0}{x}v^1_i+r^1_i|^2-|\scalT{x_0}{x}v^2_i+r^2_i|^2\right)\right)= $
%
%\textit{Theorem \ref{theo:mainFourier} states that $O(2 (\log(n))^3)$ Fourier masks ensure robust phase recovery. Robustness to noise and distortion holds like in the Gaussian case. } 
\subsection{Theoretical analysis: Correctness in Expectation and Concentration}
In this section we sketch the main steps of the proof of Theorem \ref {theo:mainFourier}. The reader is referred to Section \ref{app:Fourier}  for detailed proofs.
\\
\noindent The following proposition shows that $x_0$ is indeed  the leading  eigen-vector of the expected problem \eqref{eq:expected} with eigen-value $\lambda$, where $\lambda$ is given in \eqref{eq:lambda}. . Moreover the expected matrix $C$ is rank one:
\begin{proposition}[Correctness in Expectation]\label{pro:rankone}
The following statements hold:
 \begin{enumerate}
\item For all $x \in \mathbb{C}^n, ||x||=1$, we have the following equality,
\begin{equation}
\mathcal{E}^{x_{0}}(x)= x^*Cx=\lambda \left| \scalT{x_0}{x}\right|^2.
\end{equation}
\item Let $y=Q^{\theta}_{A}(x_0)$, $C$ is a rank one matrix,
\begin{equation} 
 C=\lambda x_0x_0^*.
\end{equation}
\item $x_0$ is an eigenvector of $C$ with  eigenvalue $\lambda$,   
\begin{equation}
Cx_0=\lambda x_0.
\end{equation}
\item The maximum eigenvector of $C$ is of the form $x_0e^{i\phi}$, where $\phi \in [0,2\pi]$.
The maximum eigenvalue is given by $\lambda$.\\
\end{enumerate}
\end{proposition}
\noindent The following lemma is a comparison equality that allows us to bound $\left|\left|xx^*-x_0x_0^*\right|\right|^2_{F}$, for any point $x$, by the excess risk $\mathcal{E}^{x_0}(x_0)-\mathcal{E}^{x_0}(x)$:
\begin{lemma}
The following equality holds for all $x\in \mathbb{C}^n$:
$$\mathcal{E}^{x_0}(x_0)-\mathcal{E}^{x_0}(x)=\frac{\lambda}{2} \left|\left|xx^*-x_0x_0^*\right|\right|^2_{F}.$$
\label{lem:ineq}
\end{lemma}
\noindent The rest of the proof follows from empirical processes theory \cite{Ledoux} and concentration inequalities \cite{tropp2012user}.

%\begin{theorem}[Hoeffding inequality]
%Let $X_{i},i=\dots r $ be a sequence of independent random $n\times n$ self adjoint matrices.
%Assume that each random matrix obeys:
%$$\mathbb{E}(X_i)=0\quad \text{ and } ||X_i||\leq \Delta \text{ almost surely.}$$
%Then for all $t\geq0$,
%$$\mathbb{P}\left\{\frac{1}{r}\left|\left|\sum_{i=1}^rX_i \right|\right|\geq t\right\}\leq 2n\exp\left(-\frac{rt^2}{8\Delta^2}\right).$$ 
%In other words:
%$$\text{For } r \geq \frac{t^2}{\epsilon^2} \Delta^2\log(n), \quad \frac{1}{r}\left|\left|\sum_{i=1}^rX_i \right|\right|\leq \epsilon \text{ with probability at least } 1-n^{-t^2} . $$
%\end{theorem}

\begin{proposition}[Concentration]\label{pro:conc}
Let $$\hat{x}_r=\arg\max_{x,||x||=1}x^*\hat{C}_rx,$$
The following inequalities hold:
\begin{enumerate}
\item $$\frac{\lambda}{2}\left|\left|\hat{x}_r\hat{x}_r^*-x_0x_0^*\right|\right|^2_{F}\leq 2\left|\left|\hat{C}_r-C\right|\right|.$$
\item $$\text{For } 0<\epsilon<1,  r \geq c\frac{\log^3n}{\epsilon^2},\quad \left|\left|\hat{C}_r-C\right|\right| \leq \epsilon \text{ with probability at least } 1-O(n^{-2}).$$
\end{enumerate}
\end{proposition}
\begin{proof}[Proof of Theorem \ref {theo:mainFourier}]
The proof of Theorem  \ref {theo:mainFourier} follows form a simple combination of Proposition \ref{pro:rankone}, Lemma \ref{lem:ineq} and Proposition \ref{pro:conc}.
\end{proof}
Proofs of Corollaries \ref{theo:Noise} and \ref{cor:Distortion}  are simple consequences of  Theorem \ref{theo:mainFourier} and Lemma \ref{lem:SNR}, where we specify the value of $\lambda$ for each model.
\begin{lemma}
The values of $\lambda$ for different observation models $\theta$ are given in the following:
\begin{enumerate}
\item Noiseless setup: $\theta(z)=z $,$ \quad \lambda=1$. 
\item Noisy setup: $\theta(z)=z+\nu, \nu$ is an exponential random variable with variance $\sigma,\quad \lambda=\frac{1+2\sqrt{\sigma}}{(1+\sqrt{\sigma})^2} $. 
\item Distortion setup: $\theta(z)=\tanh(\alpha z)$, where $\alpha >0$, $\quad \lambda=\mathbb{E}\left(sign\left(1-\tanh(\alpha E_1)\tanh(\alpha E_2)\right)|E_1-E_2|\right)$ is a decreasing function in $\alpha$.
\end{enumerate}
\label{lem:SNR}
\end{lemma} 

\section{From One bit Coded Diffraction Patterns to Super-Resolution}\label{Sec:super}

We turn now to the problem of recovering a signal from its  lower end of  power spectra. 
As discussed earlier this is a problem of practical interest, as the resolution of an optical system, for instance a lens $h$  is limited by the Fraunhofer diffraction limit $f_c$.
The super-resolution factor of $h$ is therefore defined as $SRF=\frac{n}{2f_c+1}$.
In our setup the modulated signal diffracts through a lens characterized by a PSF $h$ and a cut-off frequency $f_{c}$. Hence instead of observing  the power spectra of the modulated signal $Diag(w)x_0$ we observe the power spectra of a lower resolution signal namely $h\star (Diag(w)x_0)$.
Let $\hat{u}$ be the Fourier transform of $u \in \mathbb{C}^n$, $\hat{u}=Fu$.
Note that by the properties of the Fourier transform we have:
\begin{equation}
F(h\star(Diag(w)x_0))=\hat{h}\odot \widehat{Diag(w)x_0}
\end{equation}
Hence we observe :
\begin{equation}
|F(h\star (Diag(w)x_0))|^2=|\hat{h}|^2\odot |\widehat{Diag(w)x_0}|^2
\end{equation}
In this section we re-index $k$ for convenience $-\frac{n}{2}\leq k\leq \frac{n}{2}-1$. We use also the following convention $sign(0)=1$, this choice is arbitrary.
Note that due to the diffraction limit $f_{c}$, $h$ satisfies:
$$|\hat{h}|^2_{k}=0 ,\text{ for } k \notin  \llbracket-f_c, f_c\rrbracket .$$ 
Hence our phase-less measurement are missing in high frequencies ranges:
\begin{equation}
b_{k}=|\hat{h}|^2_{k}|\widehat{Diag(w)x_0}|^2_{k}\quad k \in \llbracket-f_c, f_c\rrbracket \text{ and } b_{k}= 0 \text { elsewhere}.
\label{eq:SR}
\end{equation}
The one bit coded diffraction patterns are the defined by comparing pairs of low resolution coded diffraction patterns $(b^1,b^2)$ defined as in \eqref{eq:SR} for two independent modulations $w^1,w^2\sim \mathcal{C}\mathcal{N}(0,I_n)$. 
$$y_{k}=sign\left(b^1_k-b^2_k\right), k=1\dots n,$$
Note that $y_k= 0$ if $k \notin \llbracket-f_c, f_c\rrbracket $. For  all $k \in \llbracket-f_c, f_c\rrbracket $, we have  $(d^1_{k}=|FDiag(w^1)x_0|^2_{k},d^2_{k}=|FDiag(w^2)x_0|^2_{k})\sim (\frac{1}{n}E^1_{k},\frac{1}{n}E^2_{k})$, where $E^1_k, E^2_k \text{ iid } Exp(1)$.
We are now ready to compute the value of $\lambda$ corresponding to that observation model:
\begin{eqnarray*}
\lambda&=&\mathbb{E}(\scalT{y}{|FDiag(w^1)x_0|^2-|FDiag(w^2)x_0|^2})\\
&=& \sum_{k=-\frac{n}{2}}^{\frac{n}{2}-1} \mathbb{E}\left(sign(b^1_k-b^2_{k})(d^1_{k}-d^2_{k})\right).
\end{eqnarray*}
Note that: 
\begin{equation}
\text{For } k \in \llbracket-f_c, f_c\rrbracket \quad sign(b^1_k-b^2_{k})=sign\left(|\hat{h}|^2_{k}d^1_k-|\hat{h}|^2_{k}d^2_k\right)=sign\left(|\hat{h}|^2_k(d^1_k-d^2_k)\right)=sign(d^1_k-d^2_k).
\label{eq:magic}
\end{equation}
it follows that:
\begin{eqnarray*}
\lambda&=& \sum_{k \in \llbracket-f_c, f_c\rrbracket  }\mathbb{E}\left(sign(d^1_k-d^2_k)( d^1_k-d^2_k)\right)\\
&=& \frac{1}{n} \sum_{k \in \llbracket-f_c, f_c\rrbracket  }\mathbb{E} \left(sign(E^1_k-E^2_k)(E^1_k-E^2_k)\right)\\
&=& \frac{2f_c+1}{n} \mathbb{E}(|E^1-E^2|)\\&=&\frac{2f_c+1}{n}.\\
&=& \frac{1}{SRF}.
\end{eqnarray*}
Thus we find that $\lambda$ is the inverse of the super-resolution factor, which confirms our findings on $\lambda$ as a quality factor:
$\lambda$ is small when the cut off frequency is small.
Interestingly $\lambda$ depends only on the domain of $\hat{h}$, regardless the shape or the values of the corresponding $PSF$, and this is due to the quantization step in  \eqref{eq:magic}.
This the main reason behind the feasibility of blind deconvolution within our framework.
 Corollary \ref{cor:blur} follows simply from Theorem \ref{theo:mainFourier} setting $\lambda=\frac{1}{SRF}$.
\section{Algorithms and Computations}\label{sec:comp}
\subsection{One Bit Phase Retrieval Algorithm}
A straight forward computation of the maximum eigenvector of $\hat{C}_r$
is rather expensive.
Using the power method, and the Fast Fourier transform we get a computational complexity of $$O(n\log^4(n)).$$
Note by FFT the Fast Fourier Transform and iFFT the Inverse Fast Fourier Transform.
For the power method we need to compute $\hat{C}_ru$:
\begin{eqnarray*}
\hat{C}_ru &=&\frac{1}{r} \sum_{i=1}^r Diag(w^1_i)FDiag(y_i)F^*Diag(w^{1,*}_i)u-Diag(w^2_i)FDiag(y_i)F^*Diag(w^{2,*}_i)u  \\
&=& \frac{1}{r} \sum_{i=1}^r w^1_i\odot FFT\left(y_i\odot iFFT(w^{1,*}_i \odot u)\right)- w^2_i \odot FFT\left(y_i\odot iFFT(w^{2,*}_i \odot u)\right)
\end{eqnarray*}

For each iteration we need to compute the FFT for each pairs of modulations, that costs $O(n\log(n))$ per pair.
We have $O(\log^3n)$ pairs, hence a total of $O(n\log^4(n))$ operations per iteration of the power method.

\begin{algorithm}[H]
%\begin{box}
 \begin{algorithmic}[1]
 \Procedure{FastFourier1bitCDPPhasePower}{$\{w^1_i,w^2_i\}_{i=1\dots r},y=(y_1\dots y_r),\epsilon$}
 \State  Initialize $r_0$ at random, $j=1$.
 \While{$|| u_j-u_{j-1}|| >\epsilon$ or $j=1$}
 \State $u^1_j\gets \frac{1}{r}\sum_{i=1}^r  w^1 _{i} \odot \text{FFT}\left(y_i\odot \text{iFFT}\left(w^{1,*}_i \odot u_{j-1}\right) \right)$
  \State $u^2_j\gets \frac{1}{r}\sum_{i=1}^r  w^2 _{i} \odot \text{FFT}\left(y_i\odot \text{iFFT}\left(w^{2,*}_i \odot u_{j-1}\right) \right)$
  \State $u_{j}\gets  u^1_j- u^2_j$
 \State $\hat{\lambda} \gets ||u_j||$
 \State $r_j\gets \frac{u_j}{\hat{\lambda}}$
  \State $j\gets j+1$
 \EndWhile
 \State \textbf{return} $\left(\hat{\lambda},u\right)$ \Comment{$(\hat{\lambda},u)$ is an estimate of $(\lambda,x_0)$.}
 \EndProcedure
 \end{algorithmic}
 \caption{FastFourier1bitCDPPhasePower}
 \label{Fast1bitPhaseMaskPower}
\end{algorithm}
%\end{box}

\subsection{SubExp Initialization}

In this spirit of the initialization \emph{SubExpPhase} proposed in \cite{AM} for Gaussian measurements we propose the following initialization from coded diffraction patterns.
Let $b_{i}=|FDiag(w_i)x_0|^2, i=1\dots L$, where $w_i\sim \mathcal{C}\mathcal{N}(0,I_n)$ iid.
Define $$\hat{C}_{L}=\frac{1}{L}\sum_{i=1}^L Diag(w_i)FDiag(b_i)F^* Diag(w_i^*),$$
It is possible to show that: 
$$\mathbb{E}(\hat{C}_{L})= x_0x_0^*+I_n,$$
we omit the proof and refer the reader to Lemma 3.1 in  \cite{FMasks}  for a similar argument.
As in the case of the Gaussian measurements the sample complexity of \emph{SubExpPhase} is higher than One bit Phase Retrieval.
An inspection of the proof of one bit phase retrieval shows that $O(\log^5n /\epsilon^2)$ modulations  are needed for an $\epsilon$ recovery in \emph{SubExpPhase}(we omit the details and show that this is indeed encountered in practice in Section \ref{sec:PT}).
Let $\hat{x}_L$ be the maximum eigenvector of $\hat{C}_{L}$, $\hat{x}_L$ is a good proxy of $x_0$, hence the following algorithm that proceeds by the power method in order to find the maximum eigenvector:
\begin{algorithm}[H]
 \begin{algorithmic}[1]
 \Procedure{FastFourierSubExpCDPPhasePower}{$\{w^1_i,w^2_i\}_{i=1}^r,b=((b^1_i,b^2_i)\dots (b^2_r ,b^2_{r})),\epsilon$}
 \State  Initialize $r_0$ at random, $j=1$.
 \While{$|| u_j-u_{j-1}|| >\epsilon$ or $j=1$}
 \State $u_j\gets \frac{1}{r}\sum_{i=1}^r  w^1 _{i} \odot \text{FFT}\left(b^1_i\odot \text{iFFT}\left(w^{1,*}_i \odot u_{j-1}\right) \right)+ w^2 _{i} \odot \text{FFT}\left(b^2_i\odot \text{iFFT}\left(w^{2,*}_i\odot u_{j-1}\right) \right)$
 \State $\hat{\lambda} \gets ||u_j||$
 \State $r_j\gets \frac{u_j}{\hat{\lambda}}$
  \State $j\gets j+1$
 \EndWhile
 \State \textbf{return} $\left(\hat{\lambda},u\right)$ \Comment{$(\hat{\lambda},u)$ is an estimate of $(\lambda,x_0)$.}
 \EndProcedure
 \end{algorithmic}
 \caption{FastFourierSubExpCDPPhasePower}
 \label{FastFourierSubExpCDPPhasePower}
\end{algorithm}

\subsection{Alternating Minimization Initialized with the One Bit Solution or the solution \emph{SubExpPhase}}\label{sec:AM}

Given the one bit solution or the solution of \emph{SubExpPhase} we can refine the solution, by running the alternating minimization procedure on the actual coded diffraction patterns initialized with the one bit solution or with the solution of \emph{SubExpPhase} as follows:
\begin{algorithm}[H]
 \begin{algorithmic}[1]
 \Procedure{AltMinPhase}{${w^1_1,w^2_2,\dots w^1_r ,w^2_{r}},b=(\sqrt{b^1_1},\sqrt{b^2_1}\dots \sqrt{b^1_r},\sqrt{b^2_r}),\epsilon$}
 \State \textbf {Initialize} $x \gets$ \textsc{FastFourier1bitCDPPhasePower}{($\{w^1_i,w^2_i\}_{i=1\dots r},y=(y_1\dots y_r),\epsilon$)}
 or $x\gets$ \textsc{FastFourierSubExpCDPPhasePower}{$\{w^1_i,w^2_i\}_{i=1\dots r},b=(b_1\dots b_2r),\epsilon$}
 \For{$k=1\dots t_0$} \Comment{$t_0$ is the number of iterations. }
  \State $(u^1_1,u^2_1,\dots u^1_r, u^2_{r}) \gets \left( Ph(FFT(w^1_{1}\odot x)),Ph(FFT(w^2_{1}\odot x)),\dots  \right)$
 \State  $x\gets \frac{1}{\sum_{s=1}^{r} |w^1_s|^2+|w^2_s|^2} \odot\left(\sum_{i=1}^{r}   w^{1*}_{i} \odot iFFT(\sqrt{b^1_i}\odot u^1_i)+w^{2*}_{i} \odot iFFT(\sqrt{b^2_i}\odot u^2_i)\right)$
 \EndFor
 \State \textbf{return} $x$ 
 \EndProcedure
 \end{algorithmic}
 \caption{FastAltMin+OneBitCDP initialization}
 \label{AltMinPhase}
\end{algorithm}
The computational complexity of Algorithm \ref{AltMinPhase} is $O(n\log^4n)$ by iteration.
\noindent The full analysis of this algorithm is subject to future research.

\section{Numerical Experiments }\label{num}
\subsection{One Dimensional Simulations}
In this section we test our algorithms on one dimensional simulated signals. 
We consider $x_0 \in \mathbb{R}^n$, such that $x_0$ is a gaussian vector, $x_0 \sim \mathcal{N}(0,I_n)$, we set $n=8000$.
We first study the phase  transition of One Bit Phase Retrieval and compare it to its counterpart in \emph{SubExpPhase} in Section \ref{sec:PT}. 
We then show in Section \ref{sec:Rob} the robustness of One Bit phase retrieval, to noise, distortion and blur .
 
\subsubsection{Phase Transition of One Bit Phase Retrieval Versus SubExpPhase}\label{sec:PT}
We consider in this section  how the performance of Algorithm \ref{Fast1bitPhaseMaskPower} for  \emph{1bitPhase} and  Algorithm \ref{FastFourierSubExpCDPPhasePower}  for \emph{SubExpPhase} depend on the number of measurements.  We consider $50$ trials , where  we generate  pairs of CDP $(b^1_{\ell},b^2_{\ell}), \ell=1\dots r$ according to Equation \eqref{eq:pairsCDP} and the one bit CDP $y_{\ell}, \ell=1\dots r$, according to \eqref{eq:meas}, where we set $\theta$ to be the identity (noiseless model).\\
At each trial we generate a set of new random modulation, and run Algorithms \ref{Fast1bitPhaseMaskPower}   \emph{(1bitPhase)} and  \ref{FastFourierSubExpCDPPhasePower}  \emph{(SubExpPhase)}.\\
In Figure \ref{fig:PT}, we report the empirical probability of success of each algorithm for increasing number of measurements $r$. We have a success if the error $1-|\scalT{\hat{x}_r}{x_0}|^2< \tau$. We set $n=8000,\tau=0.07$ in this experiment. 
We see that the phase transition for \emph{1bitPhase} happens earlier than the one for \emph{SubExpPhase}, which confirms that One Bit Phase retrieval allows lower sample complexity for a given precision.
\begin{figure}[]
\begin{center}
\includegraphics[width=0.45\textwidth]{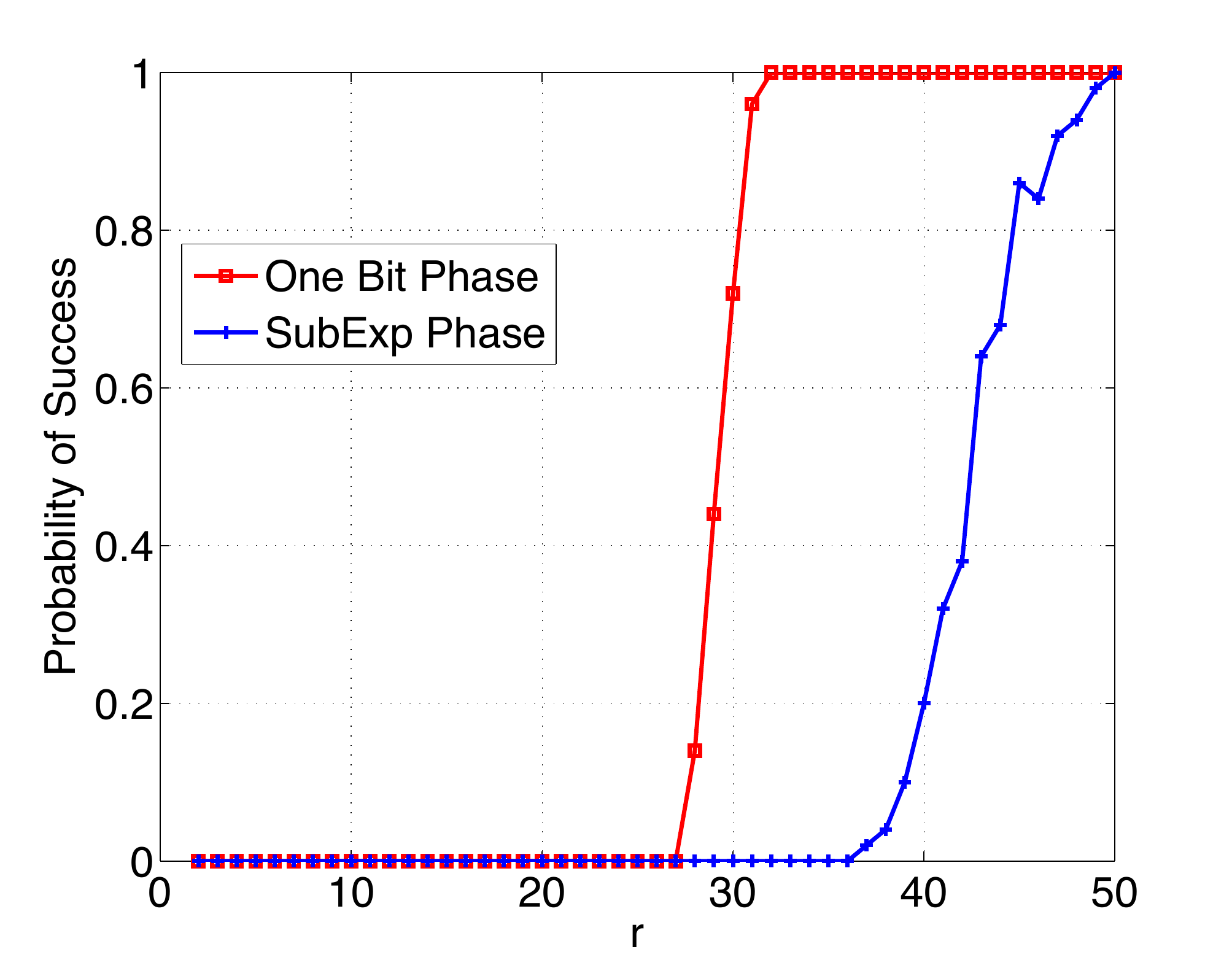}
\end{center}
\caption{Phase transition comparison of one bit Phase retrieval and SubExp phase Retrieval. }
\label{fig:PT}
\end{figure}
\subsubsection{Robustness }\label{sec:Rob}
We test the robustness of Algorithm \ref{Fast1bitPhaseMaskPower} to noise, distortion, and blur.
For this end we generate measurements according to the noise model given in Equation  \eqref{eq:1bitCDPNoise}, for increasing noise level and   for $r=10$ and $r=20$. 
We see in Figure \ref{fig:subfigNoise} that that the  recovery error of One Bit Phase Retrieval increases gracefully with the level of noise and as more measurements are available the error of recovery drops down.  
In the setup of distorted measurements we generate CDP according to Equation \eqref{eq:1bitCDPSR}, for different distortion levels $\alpha$, for different number of measurements $r=10$, and $r=20$.
We see  in Figure \ref{fig:subfigDist} that the recovery is still possible thanks to the robustness of One Bit Phase Retrieval, despite the severe non linearity.    
Finally we test the robustness of \emph{1bitPhase} to a gaussian blur with increase aperture. We generate our CDP according to the blurry model in Equation \eqref{eq:1bitCDPSR}, for $r=10$ and $r=20$.
We see in Figure \ref{fig:subfigSup} that phase retrieval ad Super-Resolution are possible with One Bit Phase Retrieval and that the error increases gracefully also with the size of the aperture and drops as more measurements are available.  
\begin{figure}[H]
%\subfigure[blue x red $\hat{x}$, $q=2^3,m=300$]{
%\includegraphics[scale=0.15]{./fig/qarylogisticregression.jpg}
%\label{fig:subfig1}
%}
\subfigure[Robustness to Exponential noise: Error $1-|\scalT{x}{x_0}|^2$ versus $\mu$ the mean of the Exponential noise]{
\includegraphics[width=0.3\textwidth]{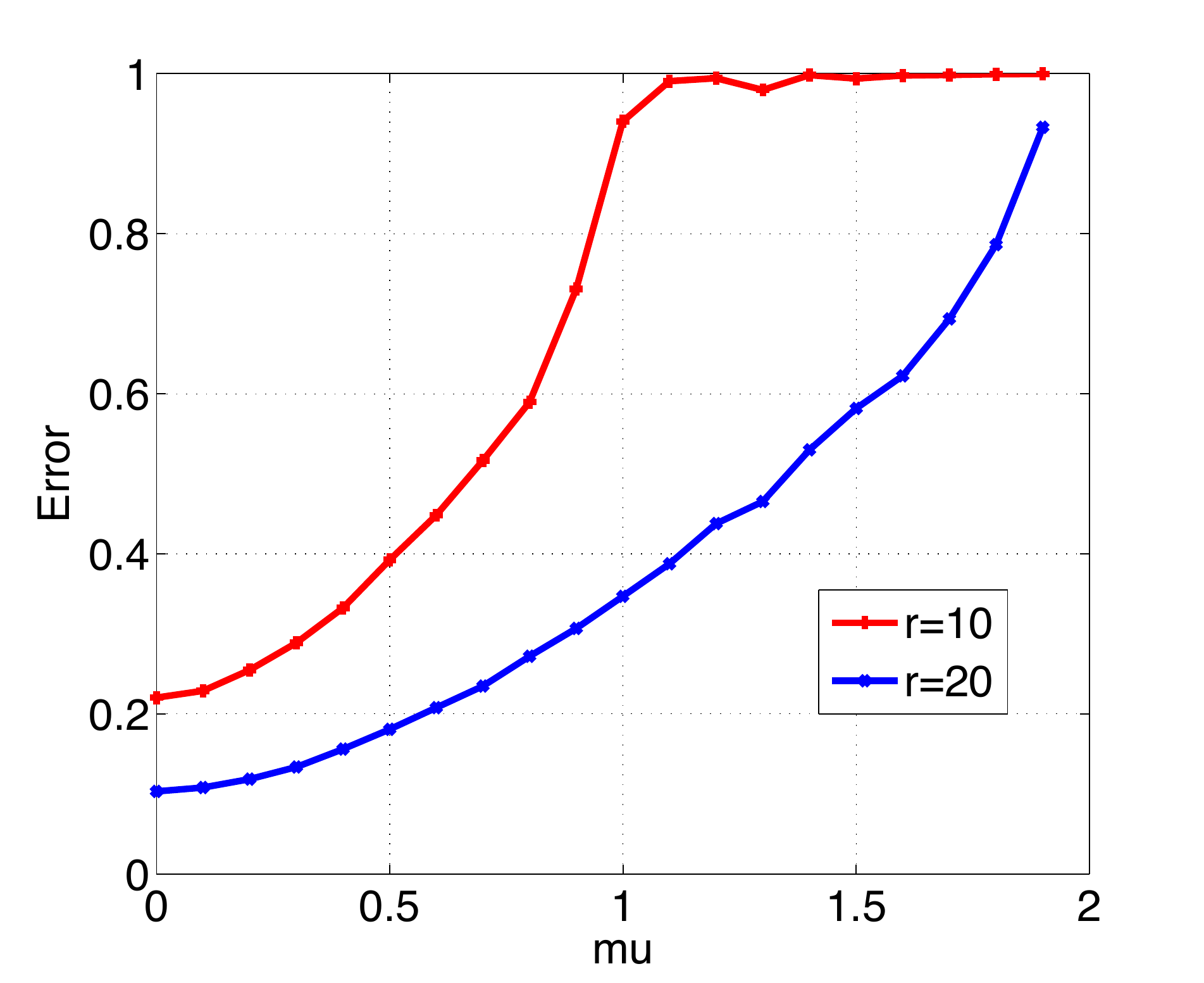}
\label{fig:subfigNoise}
}
\subfigure[Robustness to distortion: Error $1-|\scalT{x}{x_0}|^2$ versus $\alpha$ the size of the clipping.]{
\includegraphics[width=0.3\textwidth]{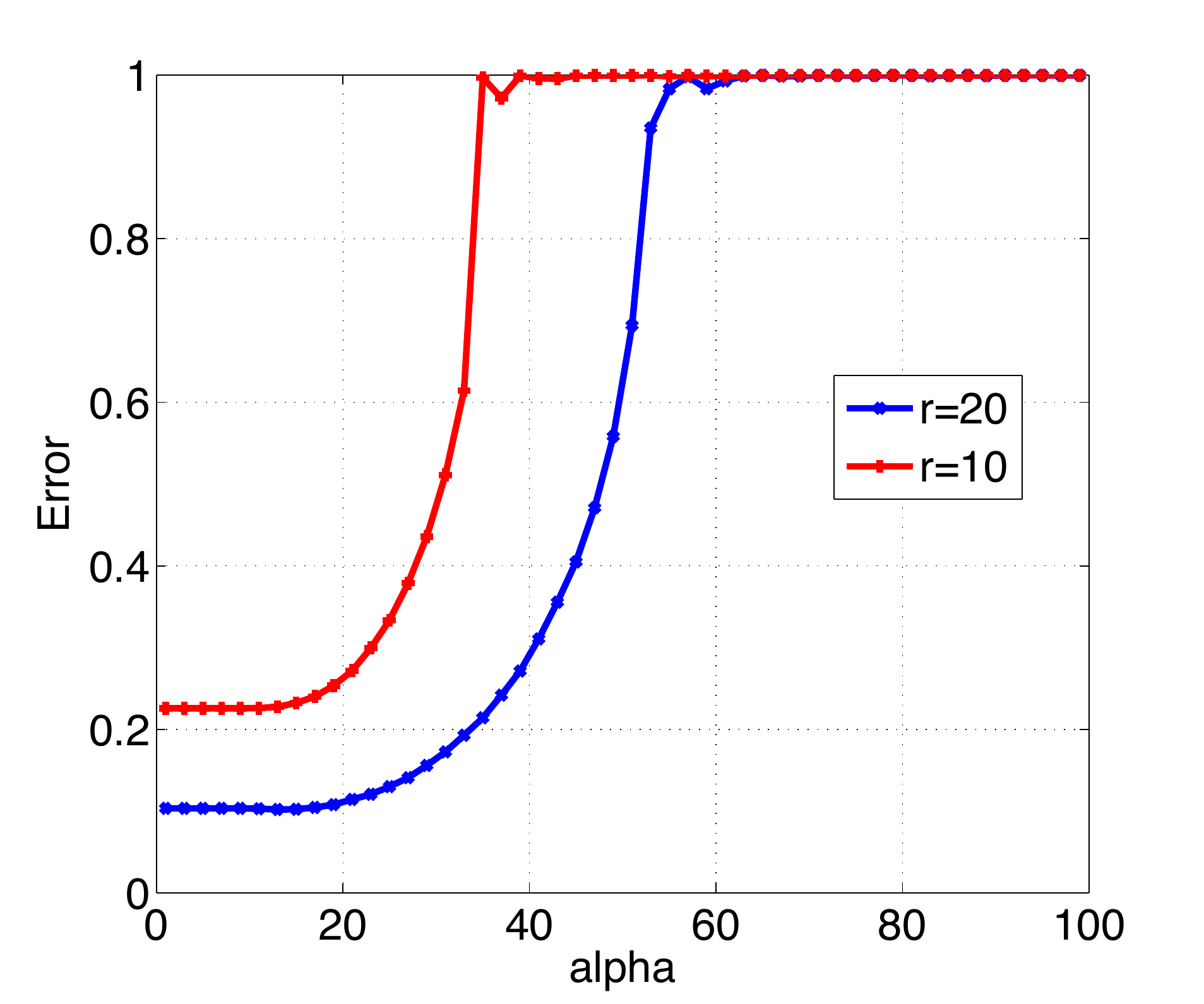}
\label{fig:subfigDist}
}  
\subfigure[Robustness to blur: Error $1-|\scalT{x}{x_0}|^2$ versus $\sigma$ the aperture of a gaussian filter.]{
\includegraphics[width=0.3\textwidth]{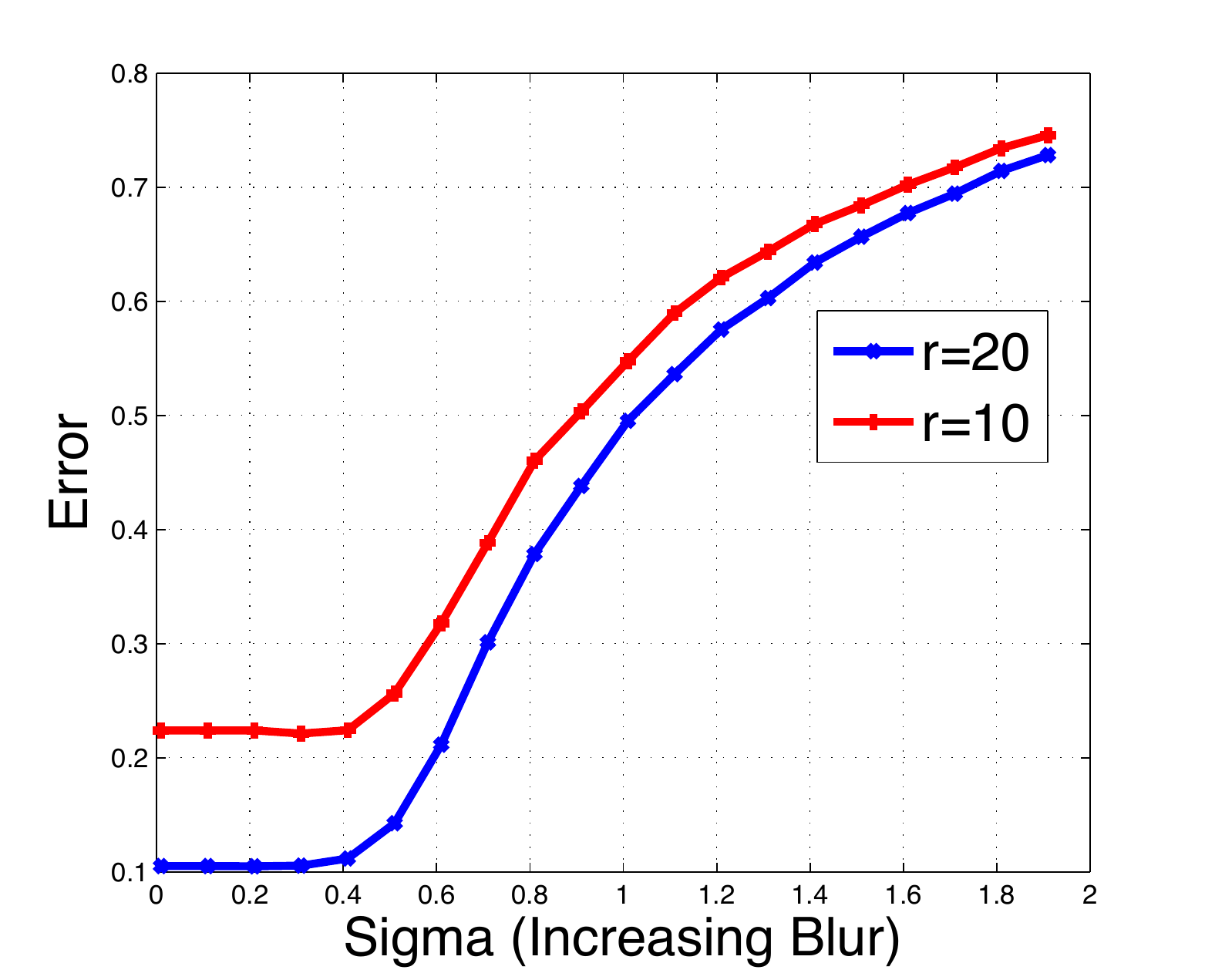}
\label{fig:subfigSup}
}  
\caption{Phase transition comparison of one bit Phase retrieval and SubExp phase Retrieval. }
\end{figure}

\subsubsection{One Bit Phase Retrieval and Alternating Minimization}
\textbf{Alternating Minimization Phase Transition.} As discussed in Section \ref{Sec:persp}, and Section \ref{sec:AM}, greedy refinements of the solution of \emph{1bitPhase} and \emph{SubExpPhase}, enhance the quality of the recovery and the sample complexity of the overall procedure as showed in \cite{AM} for the Gaussian measurements. Extending those results to Coded Diffraction patterns is subject to future work. 
We conjecture on one hand  that AM initialized with the one bit solution has a sample complexity of $O(\log^3n +\log \frac{1}{\epsilon}\log\log\frac{1}{\epsilon})$, and  on the other hand that AM initialized with \emph{SubExpPhase} solution has a sample complexity of  $O(\log^5n +\log \frac{1}{\epsilon}\log\log\frac{1}{\epsilon})$.
We show in the next section the phase transition of AM (Algorithm \ref{AltMinPhase}) initialized with \emph{1bitPhase}, \emph{SubExpPhase} and a random initialization in the noisy and the noiseless case.\\
In order to highlight the effect of the initialization step in the  AM Algorithm \ref{AltMinPhase}, we fix the number of iterations to $t_0=50$ in the noiseless and the noisy setting.\\
In figure \ref{fig:AMnonoise} we report the empirical success probability of Algorithm \ref{AltMinPhase} in the noiseless setting versus the number of pairs of modulations $r$. We declare in the noiseless case a success if $||\hat{x}_{t_0}\hat{x}_{t_0}^*-x_0x_0^*||_{F}<10^{-5}$.
We see that with that relatively small number of iterations, the phase transition of AM initialized with One Bit Phase solution happens at $r=4$. AM initialized with SubExpPhase with that limited number of iterations needs more samples to achieve phase transition at $r=10$. AM initialized at random does not achieve its  phase transition with that limited number of samples and iterations.       
This phase transition confirms the lower sample complexity of one bit solution, and its greedy refinements.\\
We now turn to the noisy setting   \eqref{eq:1bitCDPNoise} where we set $\sigma= 0.04$. In figure \ref{fig:AMnoise} we report the empirical success probability of Algorithm \ref{AltMinPhase}, where we declare in this setting a success if $||\hat{x}_{t_0}\hat{x}_{t_0}^*-x_0x_0^*||_{F}<0.03$. We see that at that accuracy level, the greedy refinements of one bit solutions are still robust to noise and superior to the other forth-mentioned  initializations (\emph{SubExpPhase} and random).
\begin{figure}[H]
%\subfigure[blue x red $\hat{x}$, $q=2^3,m=300$]{
%\includegraphics[scale=0.15]{./fig/qarylogisticregression.jpg}
%\label{fig:subfig1}
%}
\subfigure[Alternating Minimization's phase transition in the  noiseless setting.]{
\includegraphics[width=0.5\textwidth]{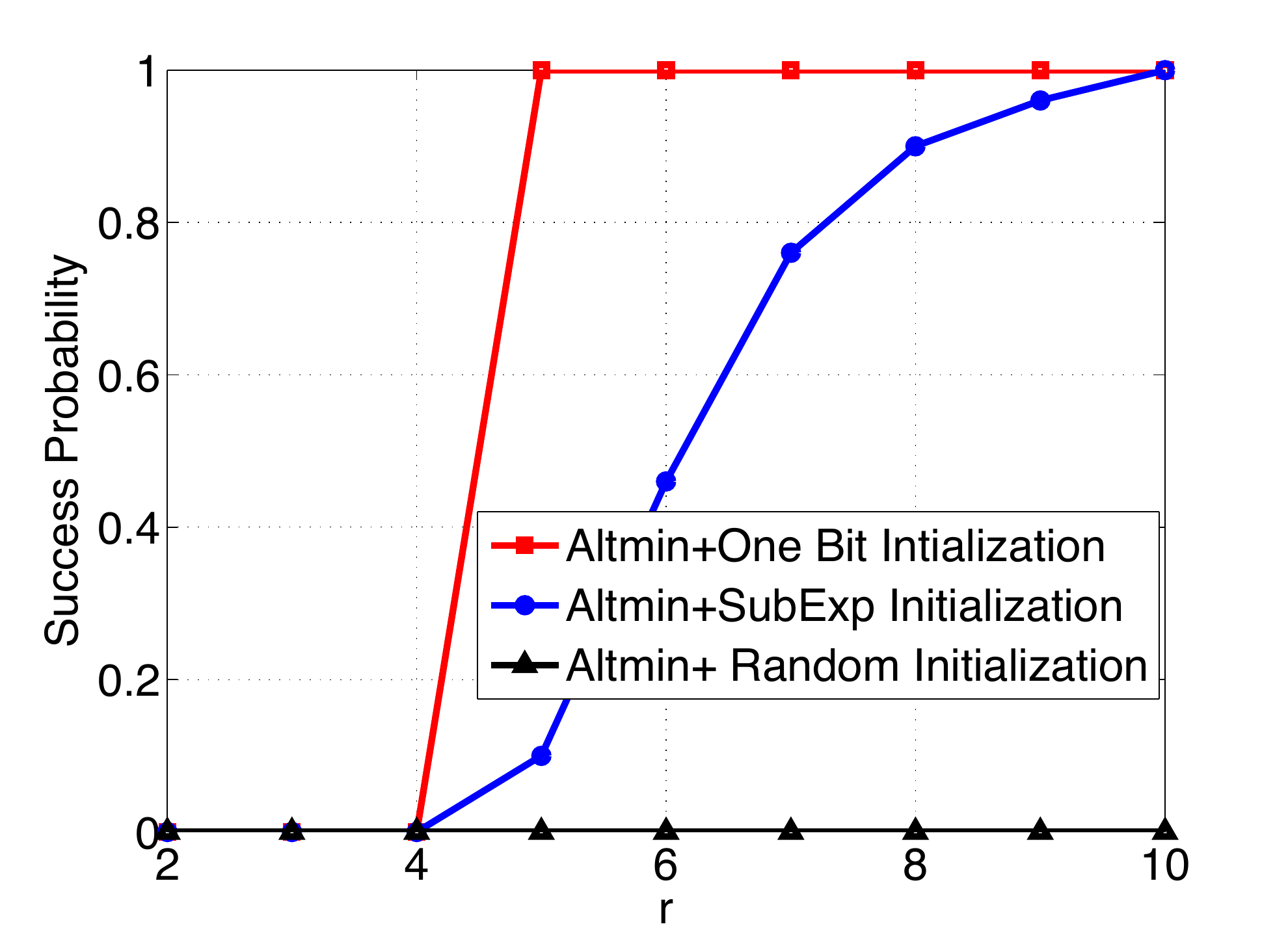}
\label{fig:AMnonoise}
}  
\subfigure[Alternating Minimization's phase transition in the noisy setting.]{
\includegraphics[width=0.5\textwidth]{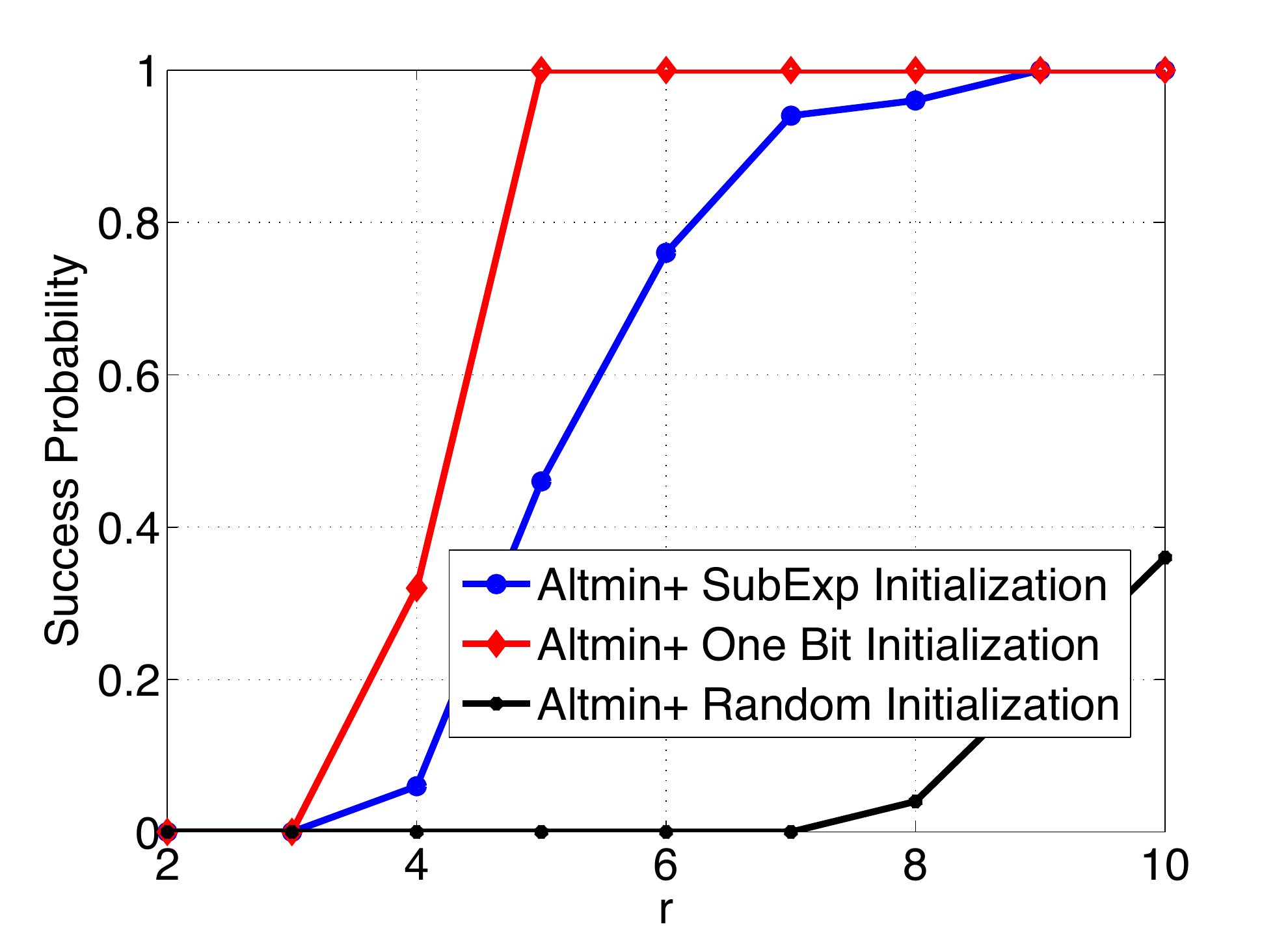}
\label{fig:AMnoise}
}
\caption{Phase transition comparison of one bit Phase retrieval and SubExp phase Retrieval. }
\end{figure}
\noindent \textbf{Error Decay.} To illustrate the benefit of the initialization step in AM we report in the following the error decay of AM with different initializations.
In figure \ref{fig:subfig2} we see that in the noiseless setting all approaches converge, the convergence is faster for AM initialized with the  one bit solution  in high dimension.
In figure \ref{fig:subfig3}\ref{fig:subfig4},\ref{fig:subfig5} we see that AM initialized with one bit solution is  more robust in the noisy setting .
\begin{figure}[H]
%\subfigure[blue x red $\hat{x}$, $q=2^3,m=300$]{
%\includegraphics[scale=0.15]{./fig/qarylogisticregression.jpg}
%\label{fig:subfig1}
%}
\subfigure[Error $1-|\scalT{x}{x_0}|^2$ versus Iterations of AltMinPhase, for  $n=8000$, and a total measurements $8n$ in the noiseless setting.]{
\includegraphics[width=0.5\textwidth]{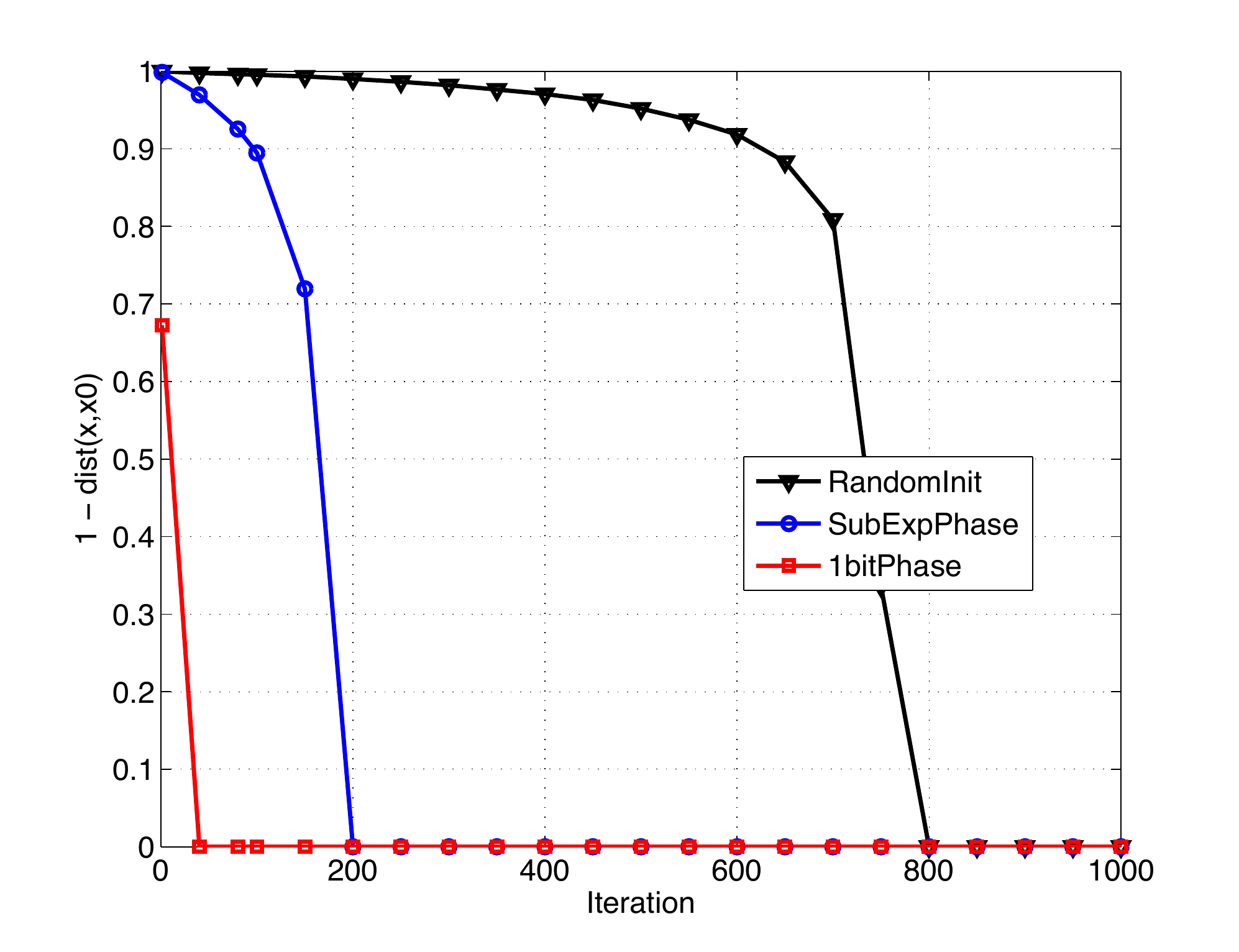}
\label{fig:subfig2}
}
\subfigure[Error $1-|\scalT{x}{x_0}|^2$ versus Iterations of AltMinPhase, for  $n=8000$ and a total measurements $8n$ in the noisy setting $\sigma=0.4$.]{
\includegraphics[width=0.5\textwidth]{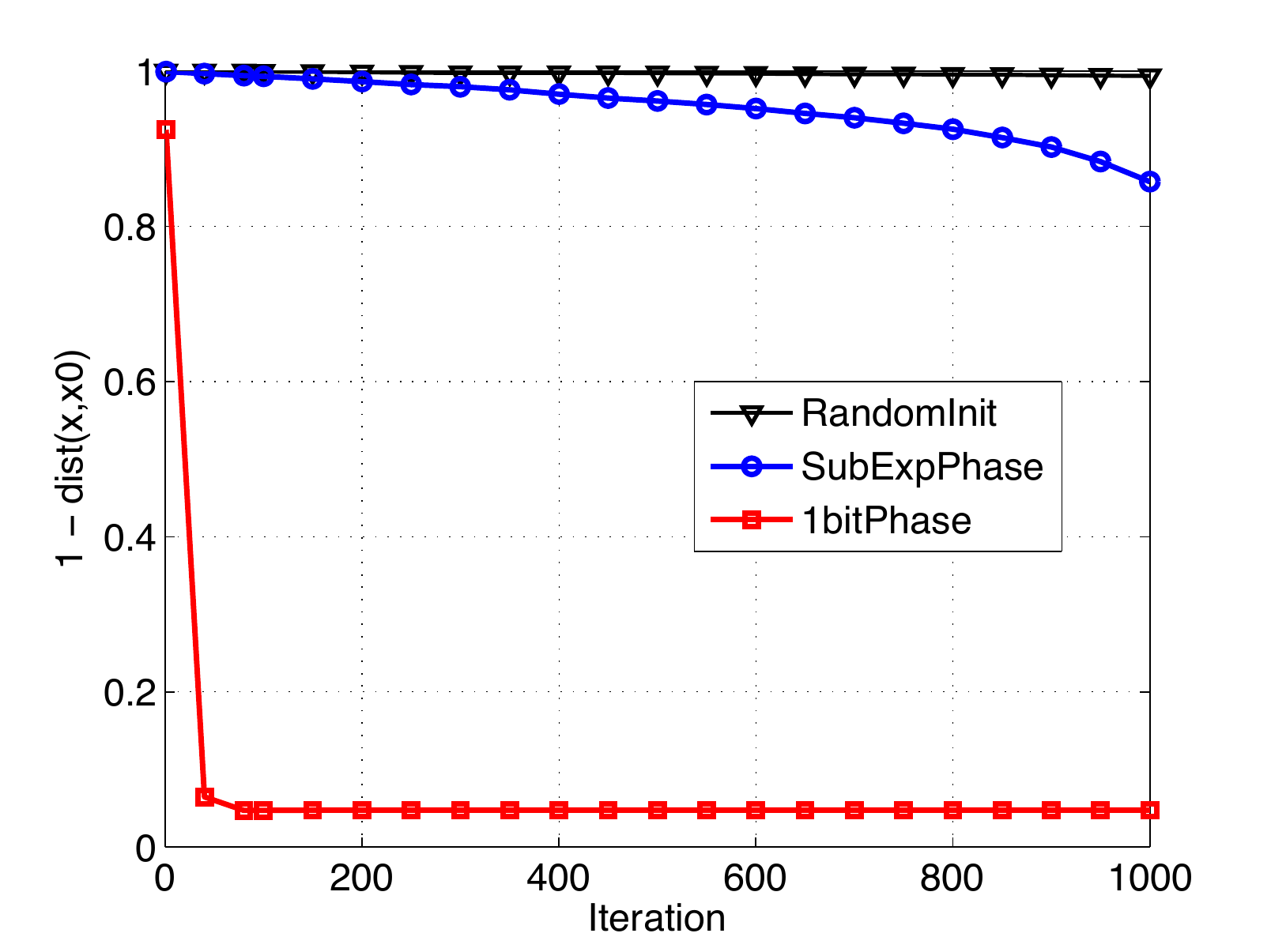}
\label{fig:subfig3}
}  
\subfigure[Error $1-|\scalT{x}{x_0}|^2$ versus Iterations of AltMinPhase, for  $n=8000$ and a total measurements $8n$  in the noisy setting $\sigma=0.8$.]{
\includegraphics[width=0.5\textwidth]{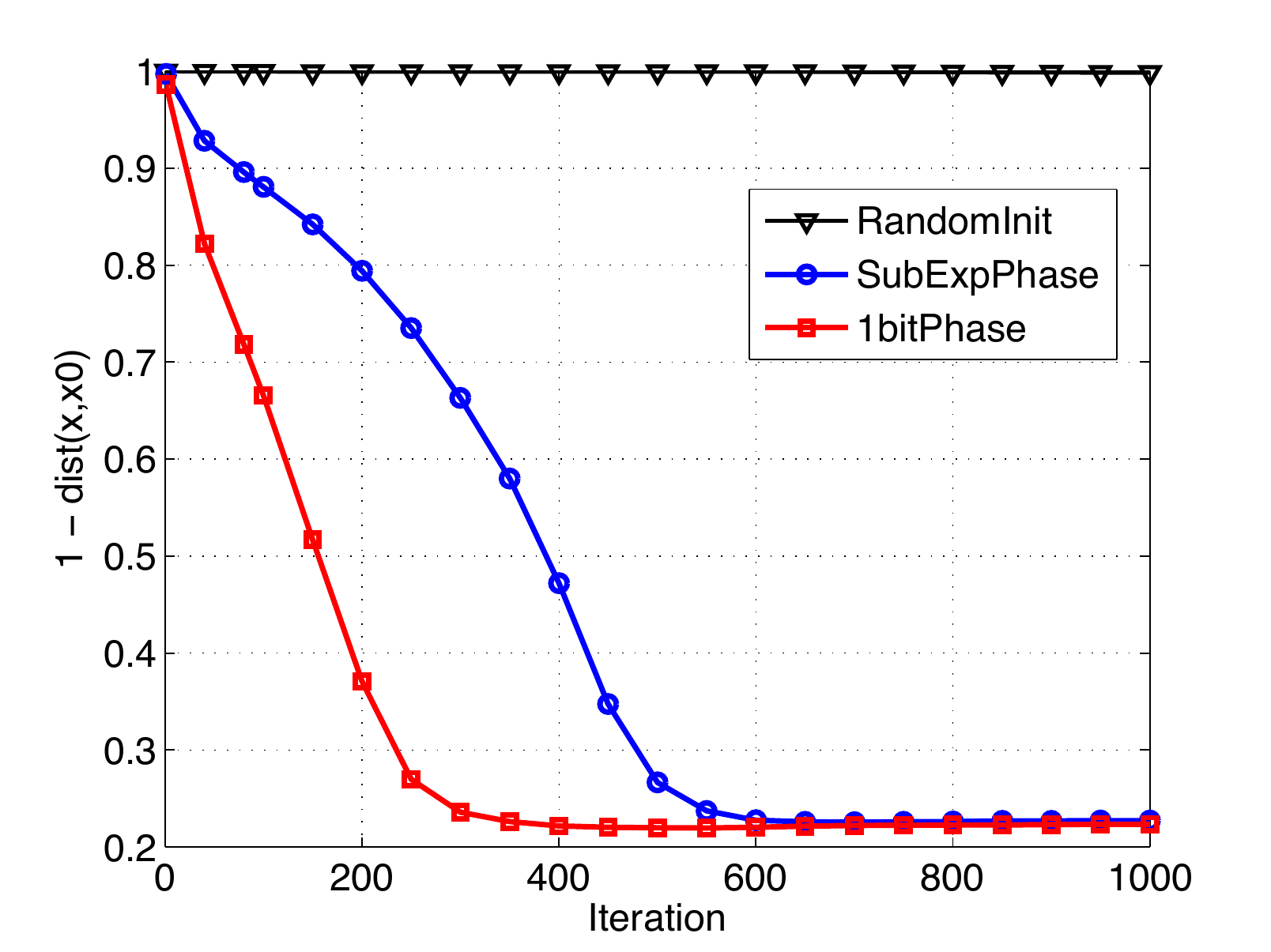}
\label{fig:subfig4}
}
\subfigure[Error $1-|\scalT{x}{x_0}|^2$ versus Iterations of AltMinPhase, for  $n=8000$ and a total measurements $8n$  in the noisy setting $\sigma=0.8$]{
\includegraphics[width=0.5\textwidth]{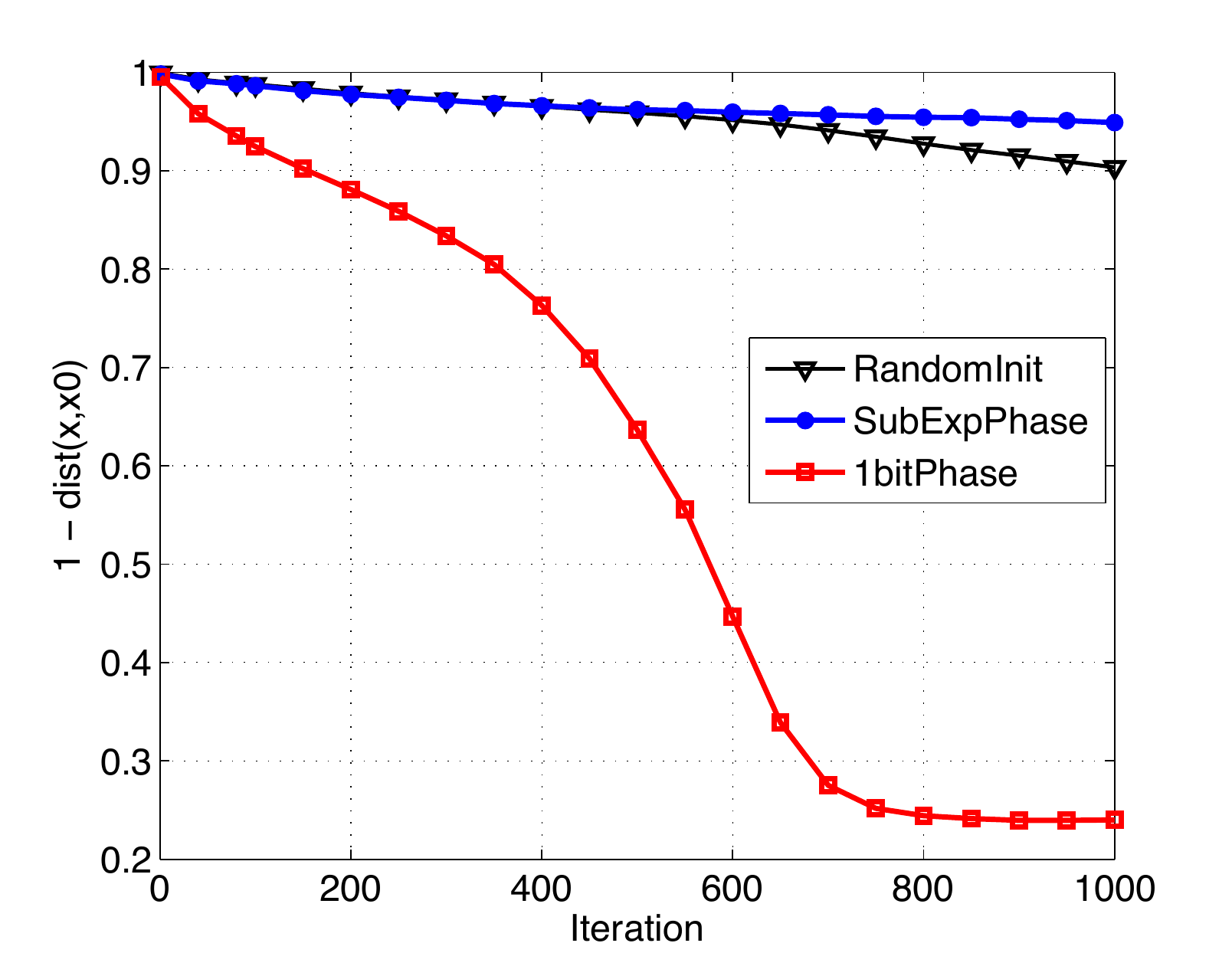}
\label{fig:subfig5}
}
\label{fig:subfigureExample}

\caption{Alternating minimization convergence with different initializations: Random Initialization,1bitPhase,and SubExpPhase, in the noisy and noiseless setting.}
\end{figure}
\noindent In the next section we test our algorithms in imaging applications, which highlights the efficiency and the robustness of the algorithms in potential applications in microscopy, astronomy and X-ray Crystallography.
%\noindent \textbf{Non Gaussian Masks}\\
\subsection{Imaging Applications}
We address in this section the problem of phase recovery in imaging applications. In the following we consider two test images, and their respective power spectra. 
The image in  Figure \ref{struct} has a dominant  edge structure, and the image in Figure \ref{text} has a dominant  textured content  .
We show that the same algorithms presented in this paper allow the phase recovery, where we  simply replace the vectors with $2D$ arrays and the $1D$ FFT with the $2D$ FFT. 
\begin{figure}[H]
\includegraphics[width=\textwidth]{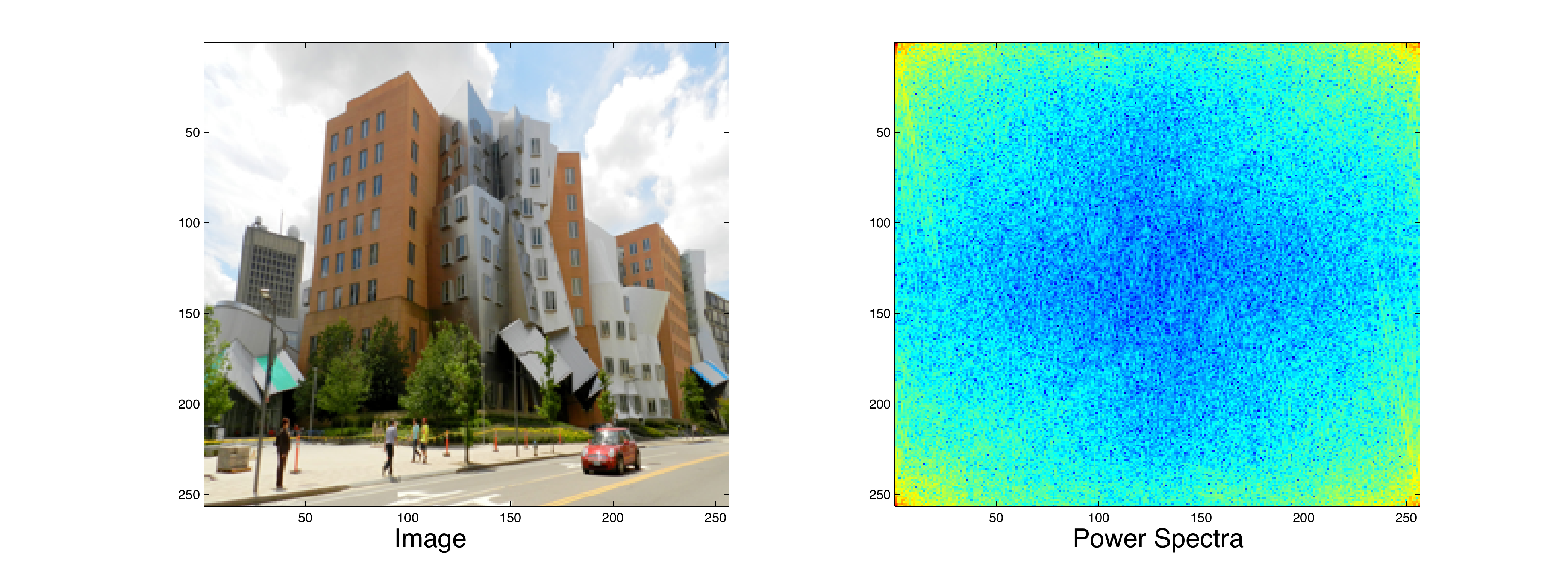}
\caption{ An image of Stata Center (Structured) and its Power Spectra (we plot the logarithm of the  power spectra  of one color Channel (R for instance) ).}
\label{struct}
\end{figure}
\begin{figure}[H]
\includegraphics[width=\textwidth]{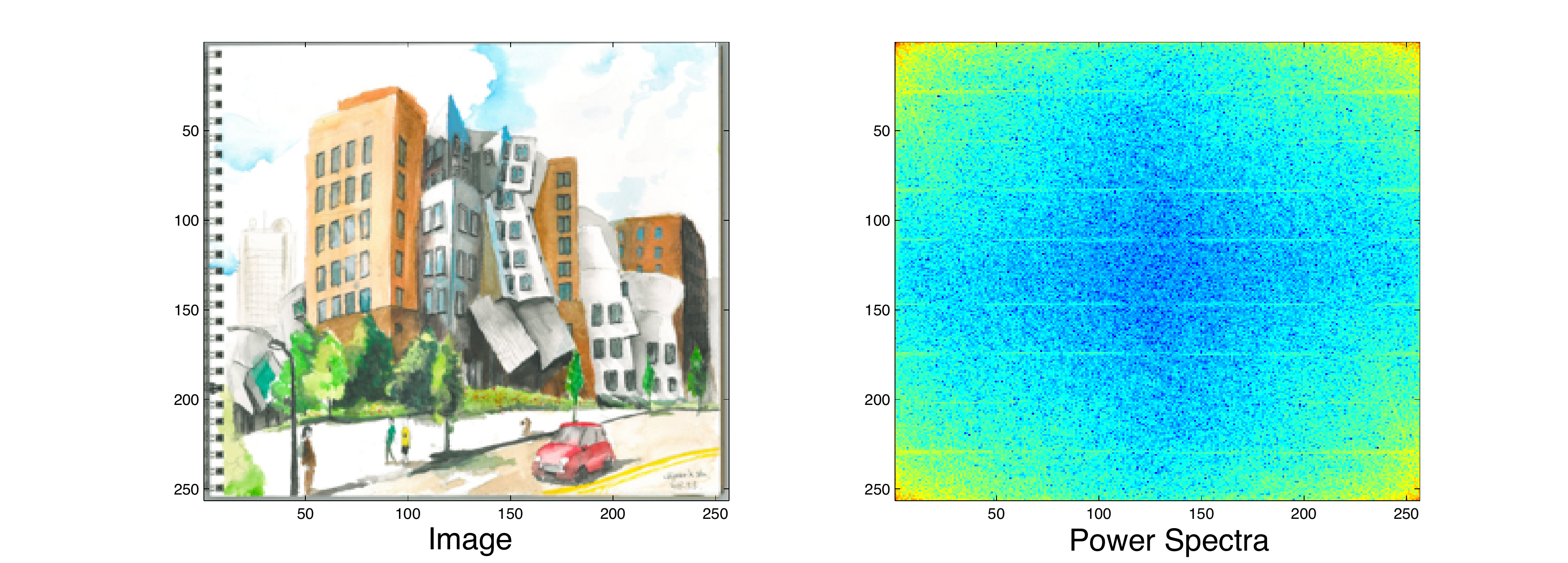}
\caption{ A drawing of Stata Center (texture  like image) and its Power Spectra (we plot the logarithm of the  power spectra  of one color Channel (R for instance) ).}
\label{text}
\end{figure}
\subsubsection{One Bit Coded Diffraction Patterns}
We modulate the image with a $2D$  Gaussian array and collect $2D$ coded diffraction patterns. 
We illustrate in Figure \ref{fig:1bitcdp2} the one bit coded diffraction patterns that become in this case a $2D$  binary array.   
The one bit array is obtained by quantizing pairs of coded diffraction arrays.
Our goal is therefore to recover robustly the image from the knowledge of One bit coded diffraction arrays using the same algorithms presented in this paper.
In the next section we test the robustness of the recovery against distortion, noise and blur.  
\begin{figure}[ht]\begin{center}
\includegraphics[width=\textwidth]{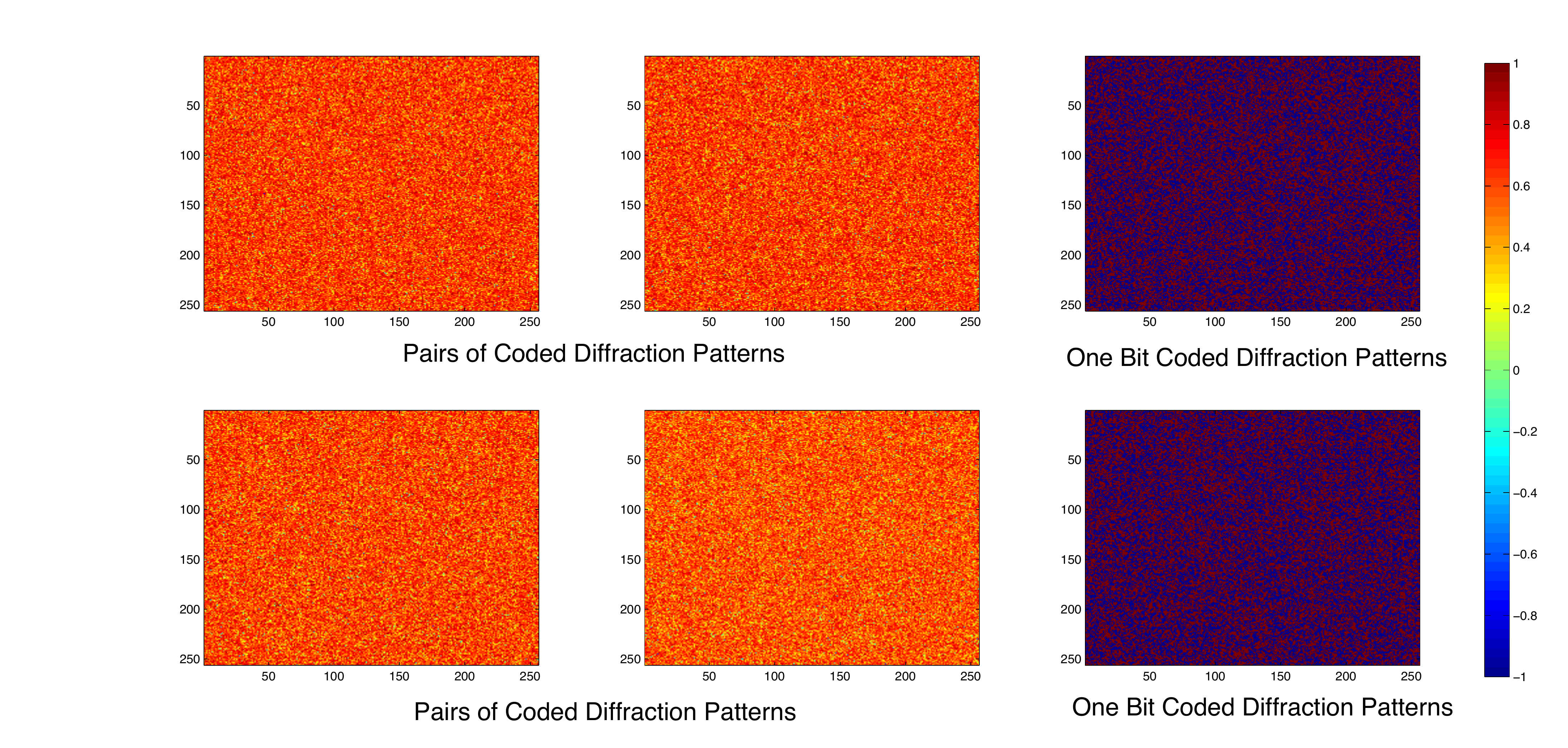}
\end{center}
\caption{ One Bit Coded Diffraction Patterns obtained by quantizing pairs of coded diffractions patterns.  }
\label{fig:1bitcdp2}
\end{figure}

\subsubsection{Phase Retrieval from Distorted Diffraction Patterns.}
In many acquisition systems we collect the power spectra of an object of interest. This acquisition might be altered by many imperfections due to multiple scattering phenomena for instance, or distortion in the precision of the CCD. Robustness to distortion such as clipping is a desirable feature in phase retrieval. In Figure \ref{fig:disto} we simulate  distorted power spectra by applying a sigmoid $(\tanh(\alpha . ))$ to the Fourier spectrum, for different clipping levels ($\alpha$). While most of the spatial frequency information is lost, phase retrieval is still possible thanks to the robustness of One Bit Phase retrieval to distortions.
We apply  $2r$ Gaussian masks to our image of interest and then collect the power spectra of each masked image. We apply to each masked power spectra  a sigmoid with a clipping parameter $\alpha$. 
We the get our One Bit CDP by quantizing pairs of distorted power spectras. In our experiment $n=256\times 256$, $r=16\lfloor\log(n)\rfloor$.
\begin{figure}[H]
\begin{center}
\includegraphics[width=1.1\textwidth]{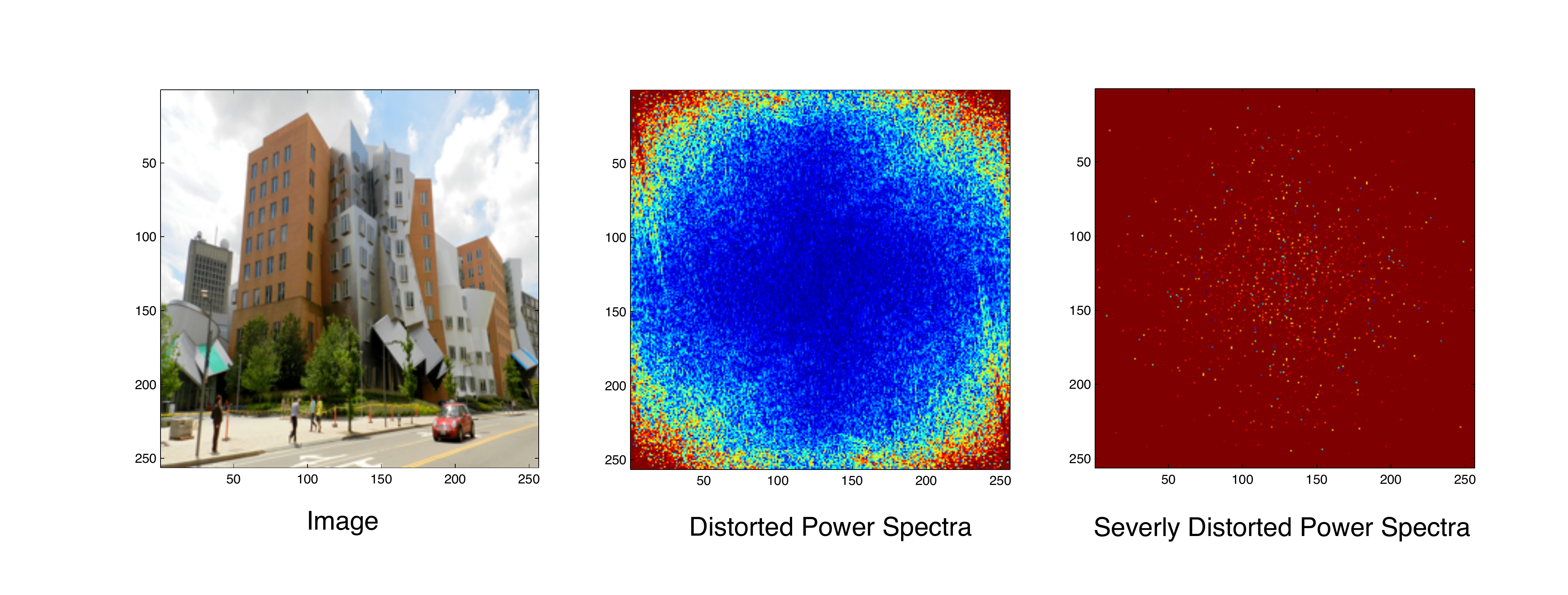}
\end{center}
\caption{Stata Center's power spectra undergoing a distortion such as clipping of the values $tanh(\alpha |\hat{x}(w)|)$. The distorted power spectra is obtained for $\alpha=0.001$. The severely distorted is obtained for $\alpha=0.1$. }
\label{fig:disto}
\end{figure}
\begin{figure}[H]
\begin{center}
\includegraphics[width=\textwidth]{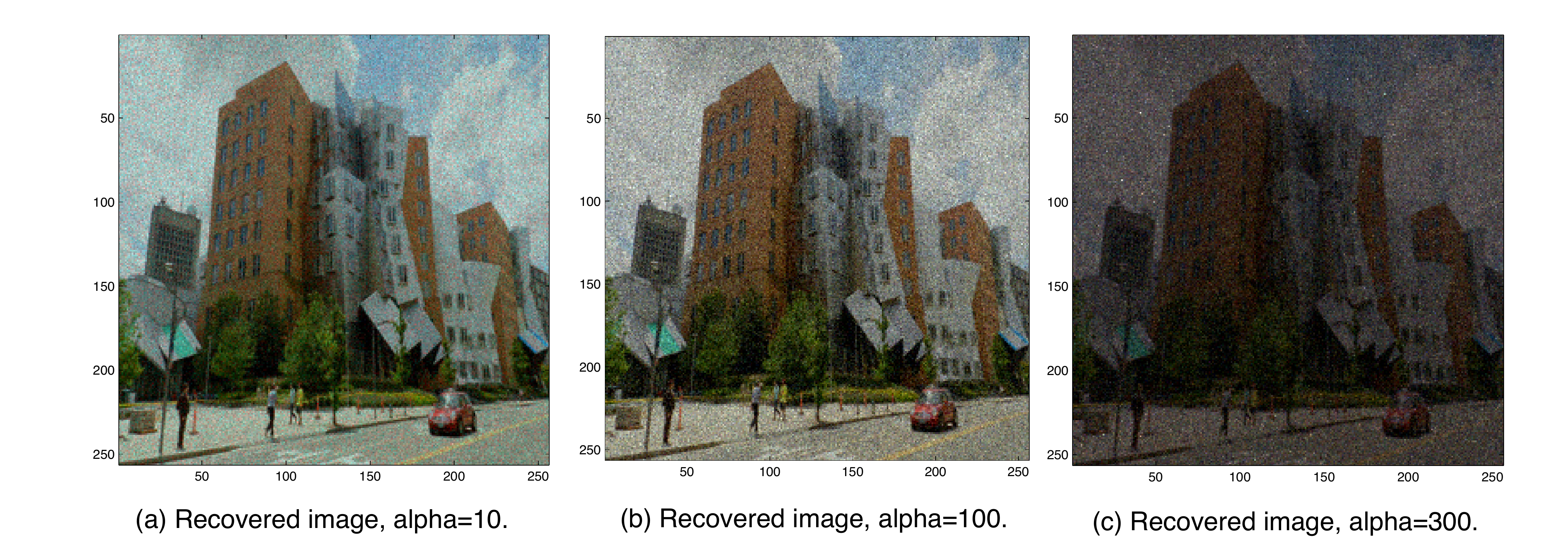}
\end{center}
\caption{One Bit Phase Retrieval from distorted CDP. (a) and (b) Recovery from small distortions. (c) Recovery from Severe Distortions. }
\label{fig:recovDistStruc}
\end{figure}
\noindent We run independently the same process of acquisition on the three color channels, as well as Algorithm \ref {Fast1bitPhaseMaskPower}. In Figures \ref{fig:recovDistStruc} and \ref{fig:recovDistTexT} we show the output of the Algorithm \ref{Fast1bitPhaseMaskPower} for various level of distortions varying from mild to sever distortion.
We see that the reconstruction in both cases is still possible from one bit CDP despite the distortion that the intensities values are undergoing.   
\begin{figure}[H]
\includegraphics[width=\textwidth]{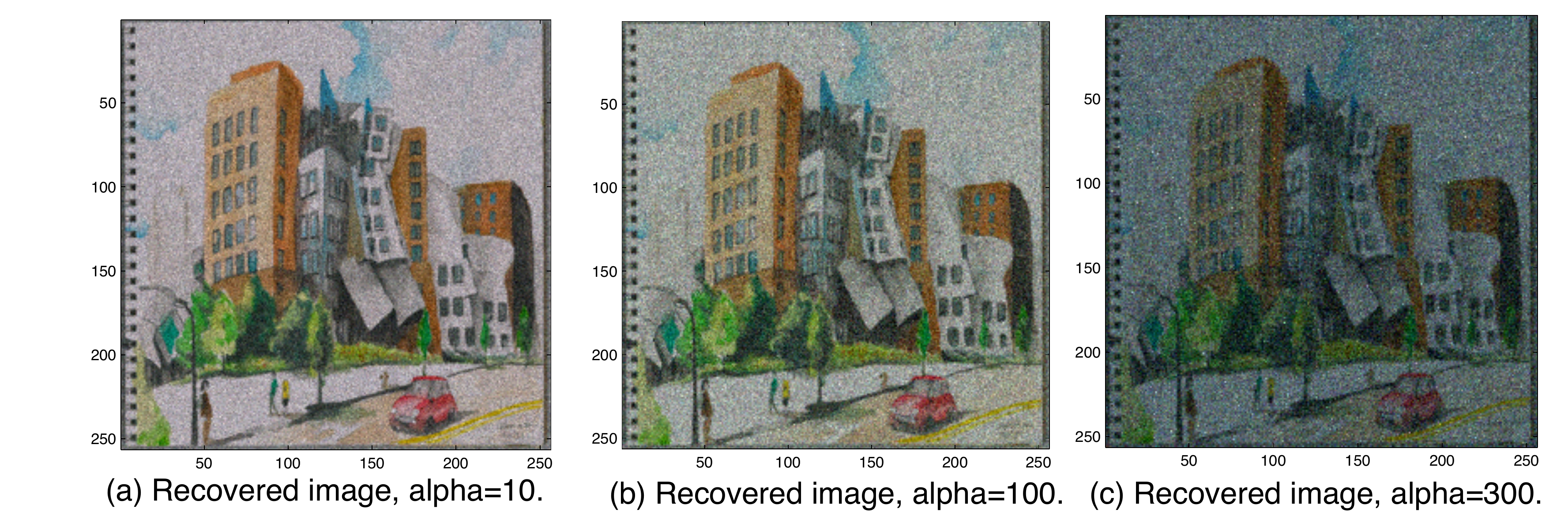}
\caption{One Bit Phase Retrieval from distorted CDP. (a) and (b) Recovery from small distortions. (c) Recovery from Severe Distortions. }
\label{fig:recovDistTexT}
\end{figure}
\subsubsection{Super-Resolution and Blind Deconvolution}

We turn now to the problem of recovering an image from its  lower end of  power spectra. 
As discussed earlier this is a problem of practical interest in microscopy, as the resolution of an optical system, for instance a lens $h$  is limited by the Fraunhofer diffraction limit $f_c$.
The super-resolution factor of $h$ is therefore defined as $SRF=\frac{n}{2f_c+1}$.
In our setup the modulated object diffracts through a lens characterized by a PSF $h$ and a cut-off frequency $f_{c}$. Hence instead of observing  the power spectra of the modulated signal $Diag(w)x_0$ we observe the power spectra of a lower resolution signal namely $h\star (Diag(w)x_0)$. We consider in this experiment $h$ to be an averaging filter. In Figure \ref{fig:SuperResBlur} we see a blurred image, obtained by convolving the original image Stata with an averaging filter of size $8\times 8$. In the following we simulate the diffraction patterns of the modulated image through the aperture $h$ by taking:$(|Fh\star (Diag(w^1_i)x_0)|^2, |Fh\star (Diag(w^2_i)x_0)|^2), i=1\dots r$. We then obtain the one bit CDP by quantizing pairs of CDPs. In our experiment we have $n=256\times256 $, $r=10\lfloor\log(n)\rfloor$.\\
In Figure \ref{fig:SuperResRecov}, we show the output of Algorithm  \ref{Fast1bitPhaseMaskPower} given the one bit CDP collected as mentioned previously. We see that most of the missing details in Figure \ref{fig:SuperResBlur} are recovered. Hence one bit phase retrieval enables super-resolution and blind deconvolution as it is agnostic to the nature of the blur.
\begin{figure}[H]

\subfigure[Image convolved with an averaging filter of  size $8\times8$.]{
\includegraphics[width=0.5\textwidth]{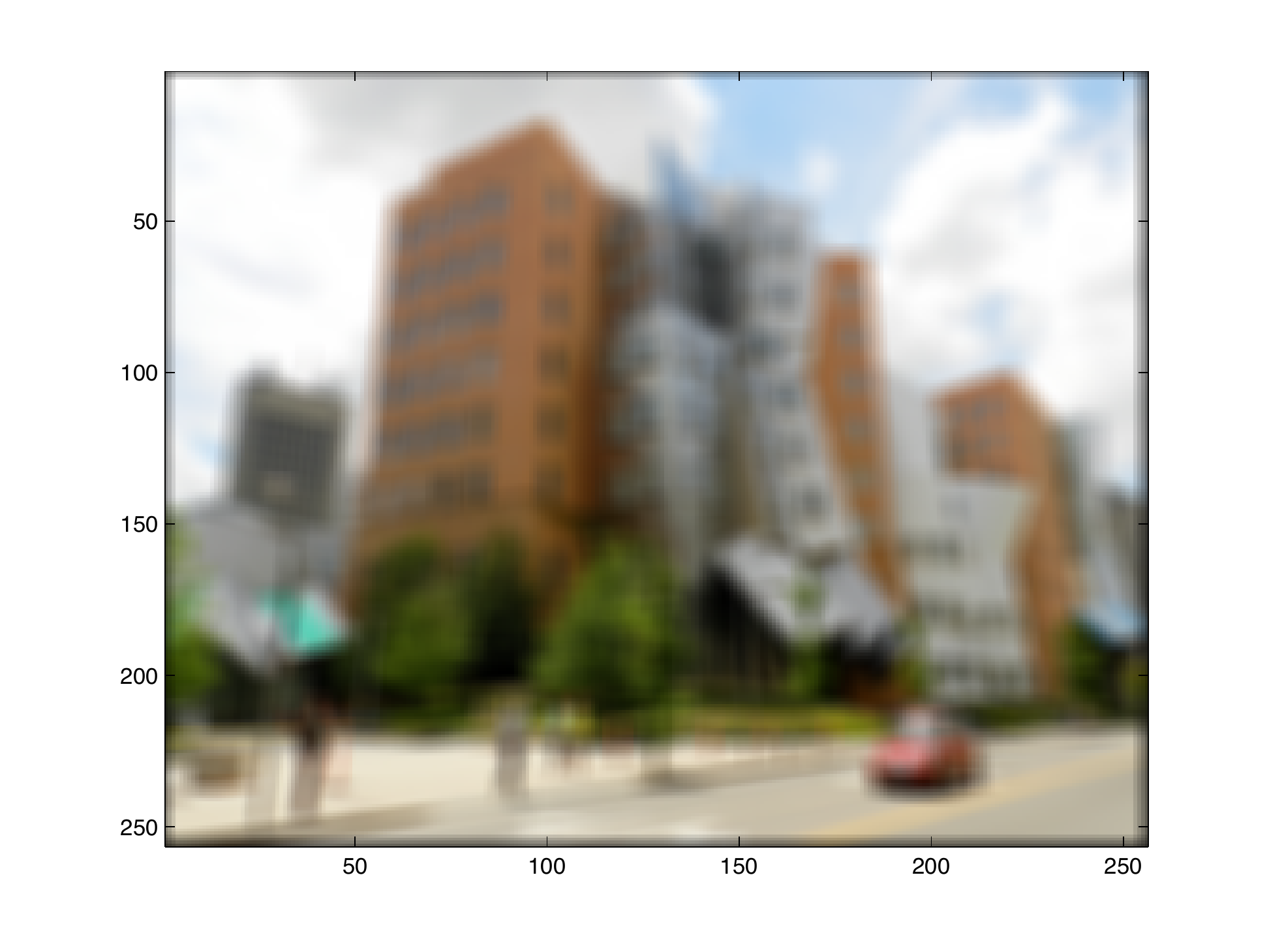}
\label{fig:SuperResBlur}
}  
\subfigure[Recovered image from One Bit CDP .]{
\includegraphics[width=0.5\textwidth]{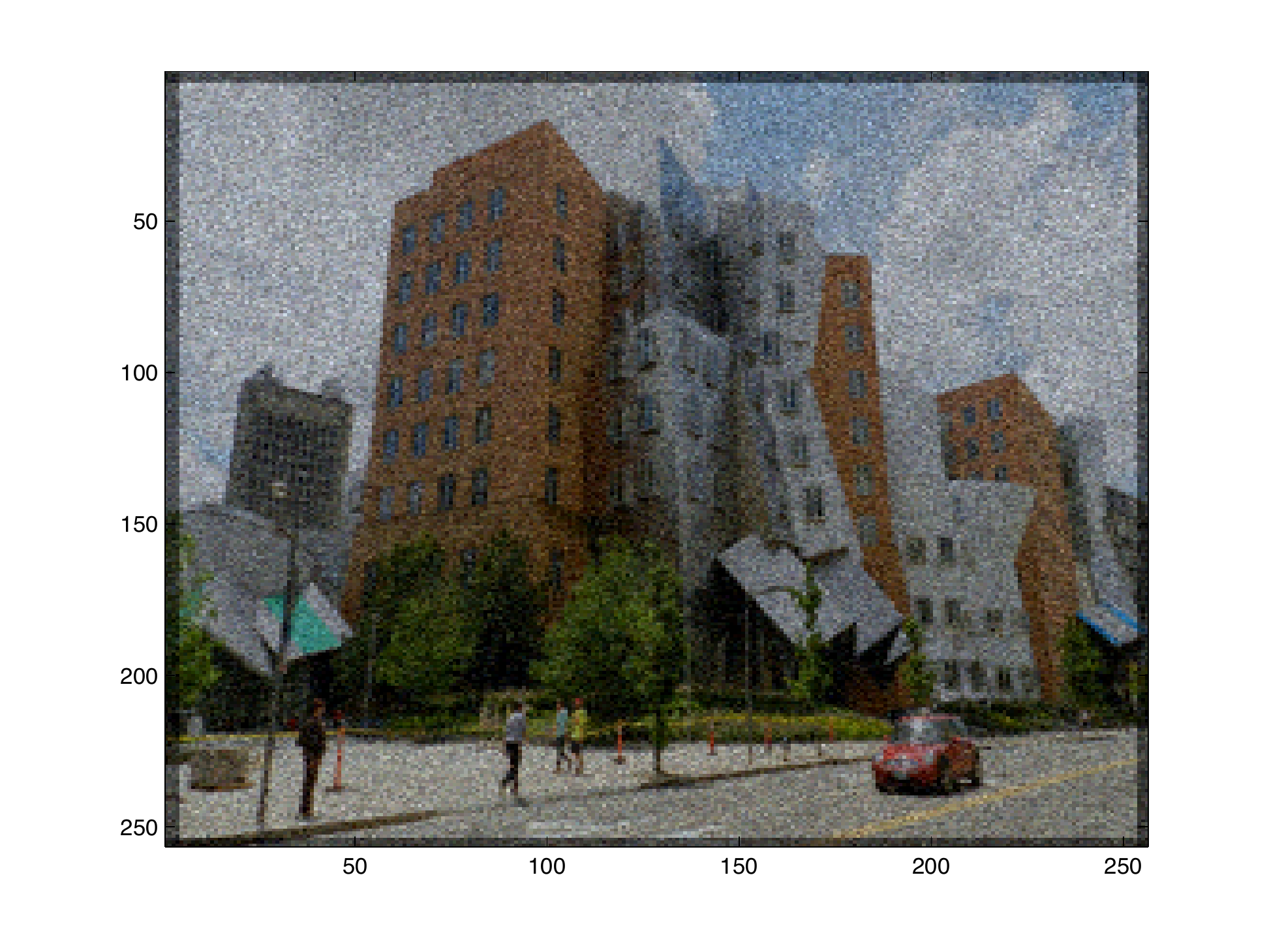}
\label{fig:SuperResRecov}
}
\caption{Super-Resolution and Blind Deconvolution via Phase Retrieval from One Bit CDP. }
\end{figure}

%\noindent \textbf{Varying PSF: Speckle Imaging}\\
%\begin{figure}[H]
%\begin{center}
%\includegraphics[width=\textwidth]{./fig/varyingPSF}
%\end{center}
%\caption{Stata Center's power spectra  for an acquisition  through multiple  apertures: In the first row we plot various Gaussian PSF (in space) with decreasing spatial resolution, in the second row we plot the power spectra of the image (Stata Center) convolved with the  corresponding aperture. }
%\end{figure}

%%\begin{figure}[H]
%%\begin{center}
%%\includegraphics[width=\textwidth]{./fig/speckle}
%%%\end{center}
%%\caption{Speckle imaging of Stata Center  : In the first row we plot various random speckle perturbations (in space), in the second row we plot the power spectra of the image (Stata Center) convolved with the  corresponding random perturbation.}
%%\end{figure}
\subsubsection{One Bit Phase Retrieval and Alternating Minimization }
In this section we test the alternating minimization Algorithm \ref{AltMinPhase} initialized with One Bit Phase Retrieval (Algorithm \ref{Fast1bitPhaseMaskPower}), in the noiseless and the Poisson noise model of equation \eqref{eq:1bitCDPoisson}.
To emphasis the effect of the initialization we set the number of iterations $t_0$ in Algorithm  \ref{AltMinPhase} to a relatively small number. In our experiments we set $t_0=50$, and $r=4$. We test our algorithms with Complex Gaussian masks and  Bernoulli masks.\\

\noindent \textbf{Gaussian masks.} We start  by the recovery for a noiseless acquisition of the CDP, the recovered image with alternating minimization initialized with one bit phase retrieval is given in Figure \ref{fig:GAMnoiseless} and is indistinguishable from the original. The average SNR on the three color channels is $101.2248 \text{ dB}$.
When the CDP are contaminated with a poisson noise (In  Equation \eqref{eq:1bitCDPoisson} we set $\eta=0.1$ ), Alternating minimization initialized with one bit phase retrieval succeeds and produces a solution with an average SNR on the three color channels of $95.625 \text{ dB}$. The recovered image in this setting is shown Figure \ref{fig:GAMnoise}.\\

\noindent \textbf{Bernoulli masks.} We repeat the same experiment with  bernoulli masks , i.e each entry of the mask is  a bernoulli random variable with parameter $p=0.8$.
The Recovered image in the noiseless is given in Figure \ref{fig:Bnoiseless}, the average SNR is $97 \text{ dB}$. In the poisson noise setting  the recovered image is given in Figure \ref{fig:Bnoise}, the average SNR is 
$95.0769 \text{ dB}$. We see that the quality for this fixed number of masks and iterations is lower than the Gaussian case. 
\begin{figure}[H]
\subfigure[Recovered Image with AM initialized with the \emph{1bitPhase} solution, for a noiseless acquisition of CDP.]{
\includegraphics[width=0.5\textwidth]{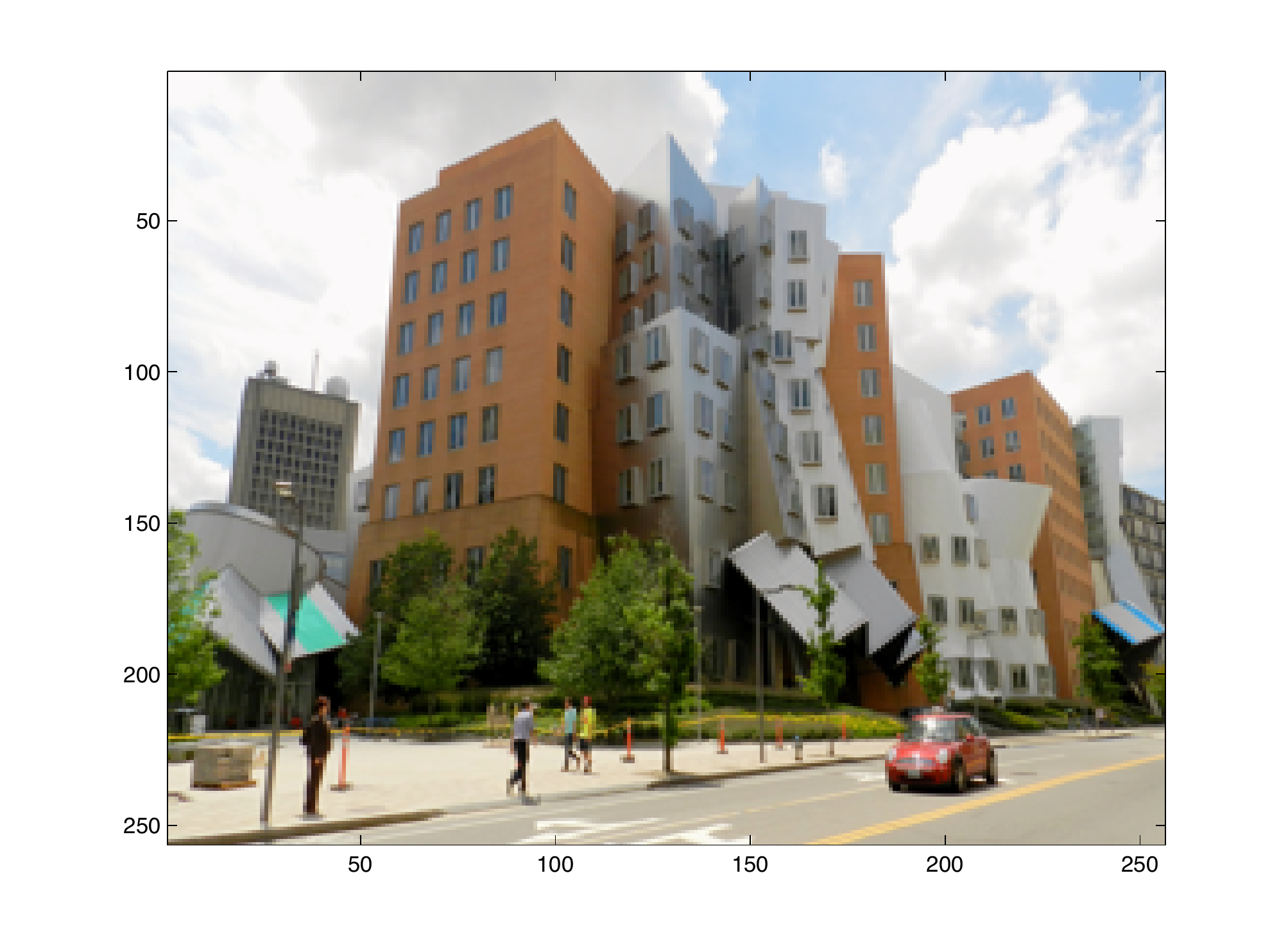}
\label{fig:GAMnoiseless}
}  
\subfigure[Recovered Image with AM initialized with the \emph{1bitPhase} solution, for an acquisition  where  the CDP are contaminated with a Poisson Noise.]{
\includegraphics[width=0.5\textwidth]{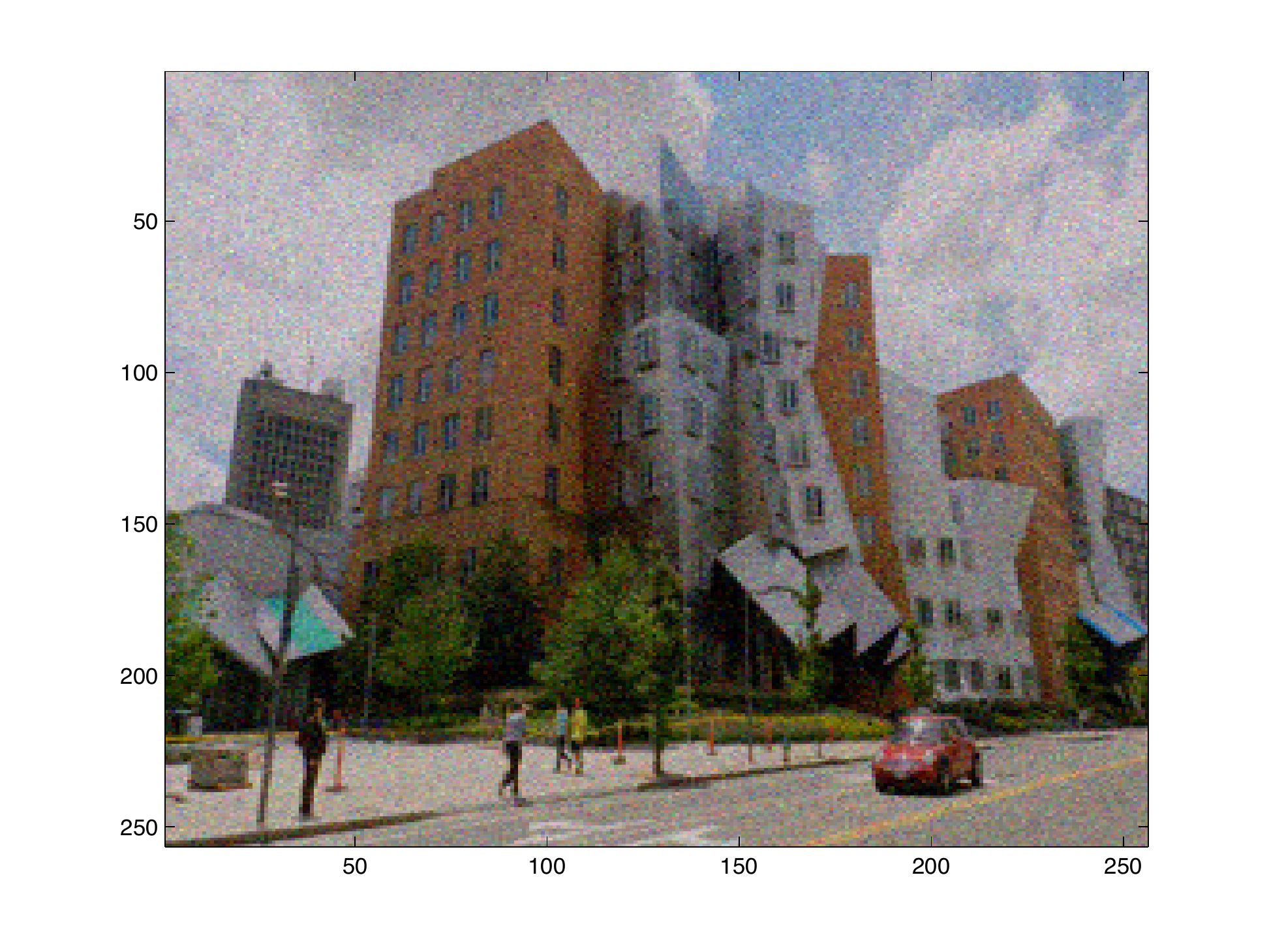}
\label{fig:GAMnoise}
}
\caption{Alternating Minimization and One Bit Phase Retrieval from Gaussian masks.}
\end{figure}

\begin{figure}[H]
\subfigure[Recovered Image with AM initialized with the \emph{1bitPhase} solution, for a noiseless acquisition of CDP.]{
\includegraphics[width=0.5\textwidth]{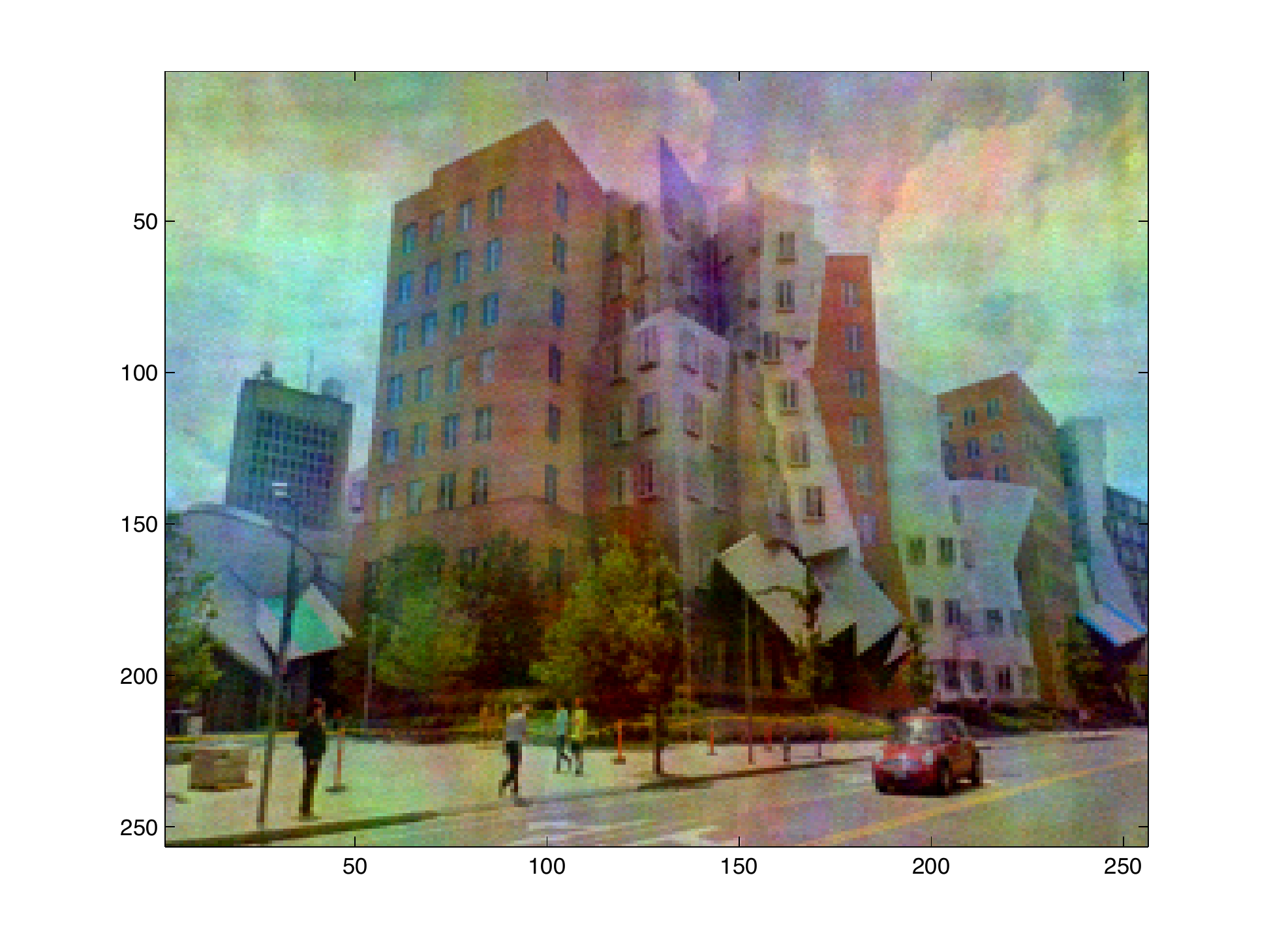}
\label{fig:Bnoiseless}
}  
\subfigure[Recovered Image with AM initialized with the \emph{1bitPhase} solution, for an acquisition  where  the CDP are contaminated with a Poisson Noise.]{
\includegraphics[width=0.5\textwidth]{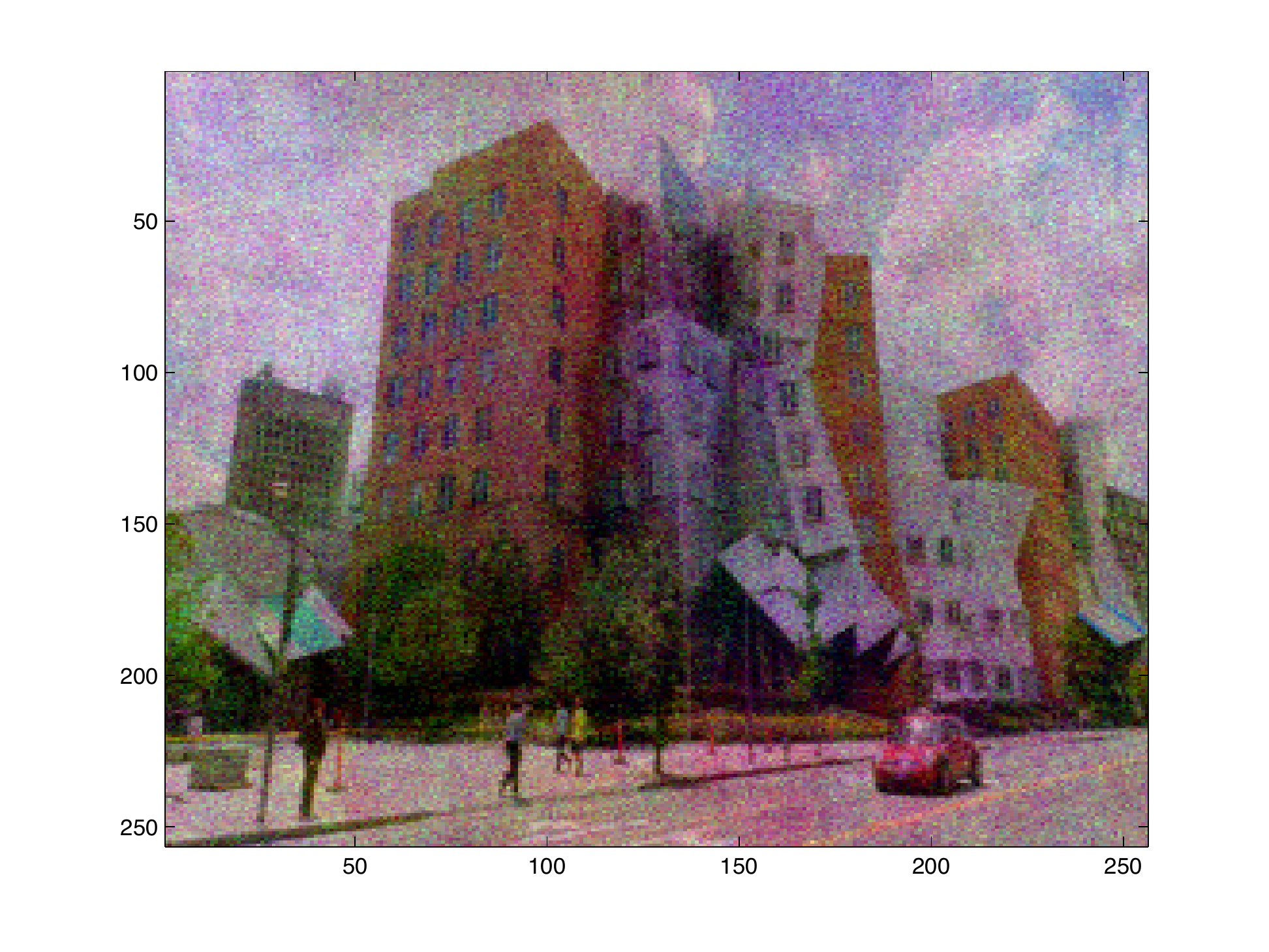}
\label{fig:Bnoise}
}
\caption{Alternating Minimization and One Bit Phase Retrieval from Bernoulli masks.}
\end{figure}

\section{Acknowledgements}
I  am thankful for fruitful discussions with  Tomaso Poggio, Lorenzo Rosasco and  Gadi Geiger.  
I would like also to thank  Gordon Wetzstein for pointing out  reference \cite{nature} 
 and for inspiring discussions on Super-Resolution and applications in microscopy.
I am also thankful to Chiyuan Zhang for providing his drawing of Stata center. 
\section{Proofs}\label{app:Fourier}
In this section we give the proofs of Proposition \ref{pro:rankone}, Lemma \ref{lem:ineq} and Proposition \ref{pro:conc}. We start with the following observation:

%\noindent \textbf{Preliminary Matrix Notation:} Let $M \in \mathbb{C}^{n\times n}$ be a complex matrix , $diag(M)$ is a vector in $\mathbb{C}^n$, containing the diagonal of $M$.\\
%Let $u\in \mathbb{C}^n$ be a complex vector, $Diag(u)$ is a matrix in $\mathbb{C}^{n\times n}$, with $u$ on the diagonal and zeros elsewhere.
\begin{lemma}\label{lem:manip}
Let $y \in \mathbb{C}^n$, and $M\in \mathbb{C}^{n\times n}$, we have the following equality: 
\begin{equation}
\scalT{y}{diag(M)}=Tr(Diag(y)^*M).
\end{equation}
We need this result for $y\in \mathbb{R}^n$:
\begin{equation}
\scalT{y}{diag(M)}=Tr(Diag(y)M).
\end{equation}
\end{lemma}

\noindent Let $v=F^* Diag(w^*)x$ and $u=F^* Diag(w^*)x_0$.
The proof of proposition \ref{pro:QForm} follows from Lemmas \ref{lem:Matrix} and \ref{lem:rewrite}:\\
\begin{lemma} The modulus vector can be rewritten in the following way:
$$|FDiag(w)x|^2 = diag\left\{  F^* Diag(w^*)xx^*Diag(w)F \right\}=diag(vv^*).$$
\label{lem:Matrix}
\end{lemma}
\begin{lemma}
The following equation holds :
$$\scalT{y}{|FDiag(w)x|^2}=  x^* Diag(w) F Diag(y)F^* Diag(w^*) x.$$
\label{lem:rewrite}
\end{lemma}
\begin{proof}[Proof of Proposition \ref{pro:QForm}]
By Lemma \ref{lem:rewrite}
\begin{eqnarray*}
\mathcal{E}^{x_0}(x)&=&\mathbb{E}(\scalT{y}{|F Diag(w^1)x|^2-|F Diag(w^2)x|^2})\\
&=& x^* C x.
\end{eqnarray*}
where $ C=\mathbb{E}\left(Diag(w^1) F Diag(y)F^* Diag(w^{1,*}) -Diag(w^2) F Diag(y)F^* Diag(w^{2,*})\right)$.
\end{proof}

\begin{proof}[Proof of Lemma \ref{lem:Matrix} ]
\begin{equation}
|FDiag(w)x|^2[i]=\sum_{j} F_{ij}w_j x_{j} \sum_{k} \bar{F}_{ik} \bar{w_k} \bar{x}_{k} = \sum_{jk} F_{ij}\bar{F}_{ik} \bar{w_k} w_j x_{j}\bar{x}_{k}.
\end{equation}
On the other hand:
\begin{equation}
e_i^*F^* Diag(w^*)xx^*Diag(w)Fe_i= [\bar{F}_{i1} \dots \bar{F_{in}}] \mathcal{W} [F_{i1} \dots F_{in}]',
\end{equation}
where $\mathcal{W}_{jk}=w_j\bar{w}_kx_j \bar{x}_k$.
Hence :
\begin{equation}
e_i^*F^* Diag(w^*)xx^*Diag(w)Fe_i=\sum_{jk} F_{ij}\bar{F}_{ik} \bar{w_k} w_j x_{j}\bar{x}_{k}.
\end{equation}
The lemma is proved.
\end{proof}

\begin{proof}[Proof of Lemma \ref{lem:rewrite}]
Using  Lemma \ref{lem:manip} and Lemma \ref{lem:Matrix},we have:
\begin{eqnarray*}
\scalT{y}{|FDiag(w)x|^2}&=&\scalT{y}{diag\left\{  F^* Diag(w^*)xx^*Diag(w)F \right\}}\\
&=& Tr ( Diag(y) F^* Diag(w^*)xx^* Diag(w) F )\\
&=& Tr( Diag(w) F Diag(y)F^* Diag(w^*) xx^* ) \\
&=& x^* Diag(w) F Diag(y)F^* Diag(w^*) x .\\
%&=& \mathbb{E}\scalT{diag(y)}{vv^*)_{F}
\end{eqnarray*}
\end{proof}

\begin{proof}[Proof of Proposition \ref{pro:rankone}]
i-Let $v=F^* Diag(w^*)x$ and $u=F^* Diag(w^*)x_0$ , $v_i$ and $u_i \sim \mathcal{C}\mathcal{N}(0,\frac{1}{n}),i=1\dots n$. $u$ and $v$ are  gaussian vectors with dependent coordinates.   The expectation of  dot product $\scalT{u}{v}$ is given in the following:
$$\mathbb{E}\left(\scalT{u}{v}\right)=\mathbb{E}\left(x_0^*Diag(w)FF^*Diag(w^*)x\right)= x_0^*\mathbb{E}(Diag(|w|^2))x= \scalT{x_0}{x}.$$
since $FF^*=I$, and $\mathbb{E}(Diag(|w|^2))=I_n$.\\
Hence $\mathbb{E}(||u||^2)= \mathbb{E}(||v||^2)=1$, since $x$ and $x_0$ are unitary.\\
There exists a complex Gaussian random vector $r$, such that each $r_i \sim \mathcal{C}\mathcal{N}(0,\frac{1}{n})$ and   $r_i$ is independent of $u_i$, $i=1\dots n$ and:
 $$v= \scalT{x_0}{x} u+\sqrt{1-|\scalT{x}{x_0}|^2}r$$
%$$\mathbb{E}(\scalT{u}{r})= \mathbb{E}(\scalT{u}{v)})-\scalT{x_0}{x}\mathbb{E}(||u||^2)= 0, \quad \mathbb{E}(r)=0, \mathbb{E}(||r||^2)=1.$$
Let $v_1=F^* Diag(w^{1,*})x$, $u_1=F^* Diag(w^{1,*})x_0$, and $v_2=F^* Diag(w^{2,*})x$, $u_2=F^* Diag(w^{2,*})x_0$.
By lemma \ref{lem:Matrix} we have:
\begin{eqnarray*}
\mathcal{E}^{x_0}(x)&=&\mathbb{E}\left(\scalT{y}{|FDiag(w^1)x|^2-|FDiag(w^2)x|}\right)\\
&=& \mathbb{E}(\scalT{y}{diag(v_1v_1^*-v_2v_2^*)})
\end{eqnarray*}
On the other hand:
\begin{eqnarray*}
vv^*&=& ( \scalT{x_0}{x}u+\sqrt{1-|\scalT{x_0}{x}|^2}r)(\overline{\scalT{x_0}{x}}u^*+\sqrt{1-|\scalT{x_0}{x}|^2} r^*)\\
&=&  |\scalT{x_0}{x}|^2uu^*+(1-|\scalT{x}{x_0}|^2) rr^*+2\sqrt{1-|\scalT{x_0}{x}|^2}Re( \scalT{x_0}{x}ur^* )
\end{eqnarray*}
Therefore:
\begin{eqnarray*}
diag(v_1v_1^*-v_2v_2^*)&=&  |\scalT{x_0}{x}|^2 diag(u_1u^*_1-u_2u_2^*)+ (1-|\scalT{x}{x_0}|^2)diag (r_1r_1^*-r_2r_2^*) \\
&+& 2\sqrt{1-|\scalT{x_0}{x}|^2} diag (Re( \scalT{x_0}{x}(u_1r^*_1 -u_2r_2^*) )\\
&=&  |\scalT{x_0}{x}|^2 ( |F Diag(w^1)x_0|^2 - |F Diag(w^2)x_0|^2) + (1-|\scalT{x}{x_0}|^2)diag (r_1r_1^*-r_2r_2^*)\\
&+& 2\sqrt{1-|\scalT{x_0}{x}|^2} diag (Re( \scalT{x_0}{x}(u_1r^*_1 -u_2r_2^*) ).
\end{eqnarray*}
Therefore, taking the expectation we have:
\begin{eqnarray*}
\mathcal{E}^{x_0}(x)= |\scalT{x}{x_0}|^2 \mathbb{E}\left( \scalT{y}{|F Diag(w^1)x_0|^2 - |F Diag(w^2)x_0|^2}\right)
\end{eqnarray*}
since $\mathbb{E}_{u_1,u_2}\scalT{y}{\mathbb{E}_{r_1,r_2}(diag (r_1r_1^*-r_2r_2^*)|u_1,u_2)}=0$,\\ and $\mathbb{E}_{u_1,u_2}(\scalT{y}{\mathbb{E}_{r_1,r_2}(diag (Re((u_1r^*_1 -u_2r_2^*) ))|u_1,u_2)})=0$.\\
Let $$\lambda= \mathbb{E}\left( \scalT{y}{|F Diag(w^1)x_0|^2 - |F Diag(w^2)x_0|^2}\right).$$
Let $\tilde{E^1_i}= |\sum_{j=1}^n F_{ij} w_j x_{0,j}|^2$, since $w_j$ is complex Gaussian , $\sum_{j=1}^n F_{ij} w_j x_{0,j}$ is $\mathcal{C} \mathcal{N}(0,\frac{1}{n})$,
therefore $\tilde{E^1_i}$ is exponential with mean $\frac{1}{n}$ (since $|F_{ij}|^2=1$). Let $E^1_i$ be an exponential with mean 1. $\tilde{E^1_i}=\frac{1}{n}E^{1}_i$.
It follows that:
\begin{eqnarray*}
\lambda&=& \mathbb{E}\left( \scalT{y}{|F Diag(w^1)x_0|^2 - |F Diag(w^2)x_0|^2}\right)\\
&=& \frac{1}{n}\sum_{i=1}^n \mathbb{E}(y_i (E^1_i-E^2_i)\\
&=& \frac{1}{n}n \mathbb{E}(y(E^1-E^2))\\
&=& \mathbb{E}(y(E^1-E^2))
\end{eqnarray*}
 Where $\lambda$ is defined for different models in Lemma \ref{lem:SNR} .
\end{proof}
\begin{remark}
Similar  results hold if we use instead of the DFT $F$ any unitary matrix $U$.
The value of $\lambda$ would depends on $|U_{ij}|^2$.
\end{remark}

\begin{proof}[Proof of Lemma \ref{lem:SNR}]
i. \textit{Noiseless:} \\$\lambda= \mathbb{E}(sign(E_1-E_2)(E_1-E_2))=\mathbb{E}(|E_1-E_2|)=1$,
since $E_1-E_2 \sim Exp(1)$.\\
ii.\textit{Noisy:} \\Let $y=sign\left((E_1+\nu_1)- (E_2+\nu_2)\right)$.
\noindent Let $L= E_1-E_2$, $L$ follows a Laplace distribution with mean 0  and scale parameter $1$: $$L \sim Laplace(0,1).$$
Let $N= \nu_1-\nu_2$, $N$ follows a Laplace distribution, $N \sim Laplace(0,\frac{1}{\gamma})$.
\noindent It follows that:
\begin{eqnarray*}
\lambda&=&\mathbb{E}_{L,N}\left(sign(L+N)L\right)\\
&=&\mathbb{E}_{L}\left(\left(1-2\mathbb{P}_{N}(N\leq -L)\right) L\right)\\
&=& \mathbb{E}_{L}\left((1-2F_{N}(-L))L\right)\\
&=& \mathbb{E}_{L}\left\{\left(1-2\left(\frac{1}{2}+\frac{1}{2}sign(-L)\left(1-\exp(-\gamma |L|)\right)\right)\right)L\right\}\\
&=& \mathbb{E}_{L} (sign(L)(1-\exp(-\gamma |L|))L)\\
&=& \mathbb{E}_{L}  |L|(1-\exp(-\gamma |L|))\\
&=& 1- \int_{0}^{+\infty}z\exp(-\gamma z) \exp(-z)dz\\
&=& 1-\frac{1}{(1+\gamma)^2} >0. 
\end{eqnarray*}

\noindent Let $\sigma=\frac{1}{\gamma^2}$ be the variance of the exponential noise.
We conclude that:
$$\lambda=\frac{1+2\sqrt{\sigma}}{(1+\sqrt{\sigma})^2}.$$

iii. \textit{Distortion:}
\begin{eqnarray*}
y &=&sign(\tanh(\alpha E_1)-\tanh(\alpha E_2))\\
&=&sign(\tanh(\alpha(E_1-E_2)))\left(1-\tanh(\alpha E_1)\tanh(\alpha E_2)\right))\\
&=&sign(\tanh(\alpha(E_1-E_2))).sign\left(1-\tanh(\alpha E_1)\tanh(\alpha E_2)\right))\\
&=& sign(E_1-E_2)sign\left(1-\tanh(\alpha E_1)\tanh(\alpha E_2)\right))
\end{eqnarray*}

$\lambda=\mathbb{E}(y(E_1-E_2))=\mathbb{E}\left(sign\left(1-\tanh(\alpha E_1)\tanh(\alpha E_2)\right)|E_1-E_2|\right).$
\end{proof}
\begin{proof}[Proof of Lemma \ref{lem:ineq}]
For $x\in \mathbb{C}^n, ||x||=1$, let $\mathcal{E}^{x_0}(x)=x^*Cx$, and $\hat{\mathcal{E}}^{x_0}(x)=x^*\hat{C}_rx$.
$$\mathcal{E}^{x_0}(x_0)-\mathcal{E}^{x_0}(x)=\lambda-\lambda |\scalT{x_0}{x}|^2=\frac{\lambda}{2}||xx^*-x_0x_0^*||^2_{F}.$$
Let $\hat{x}_r=\argmax_{x,||x||=1} \hat{\mathcal{E}}^{x_0}(x)$, we have:
$$\mathcal{E}^{x_0}(x_0)-\mathcal{E}^{x_0}(\hat{x}_r)=
\mathcal{E}^{x_0}(x_0)-\hat{\mathcal{E}}^{x_0}(x_0)+
\hat{\mathcal{E}}^{x_0}(x_0)- \hat{\mathcal{E}}^{x_0}(\hat{x}_r)+
\hat{\mathcal{E}}^{x_0}(\hat{x}_r)-\mathcal{E}^{x_0}(\hat{x}_r)
.$$
Noticing that  the term $\hat{\mathcal{E}}^{x_0}(x_0)- \hat{\mathcal{E}}^{x_0}(\hat{x}_r)$ is  non-positive in light of the definition of $\hat x _r$, we have finally: $\mathcal{E}^{x_0}(x_0)-\mathcal{E}^{x_0}(\hat{x}_r)\leq 2 \sup_{x,||x||=1} \hat{\mathcal{E}}^{x_0}(x)-\mathcal{E}^{x_0}(x)= 2\left|\left|\hat{C}_r-C\right|\right|$.
Finally:
\begin{equation}
\frac{\lambda}{2}||\hat{x}_r\hat{x}^*_{r}-x_0x_0^*||^2_{F}=\mathcal{E}^{x_0}(x_0)-\mathcal{E}^{x_0}(\hat{x}_r) \leq 2 \sup_{x,||x||=1} \mathcal{\hat{E}}^{x_0}(x) -\mathcal{E}^{x_0}(x) =2\left|\left|\hat{C}_r-C\right|\right|
\label{eq:compa}
\end{equation}
\end{proof}

\begin{proof}[Proof of Proposition \ref{pro:conc}]
\noindent It follows that: 
$$\mathbb{E}(\hat{C}_r)=\lambda x_0x_0^* $$
Where  $\hat{C}_r=\frac{1}{r}\sum_{i=1}^r A_i$ where $A_i= Diag(w^1_i)FDiag(y_i)F^*Diag(w^{1,*}_i)-Diag(w^2_i)FDiag(y_i)F^*Diag(w^{2,*}_i)$\quad $i=1\dots r$.\\
By Lemma \ref{lem:ineq}, it is now clear that the sample complexity is governed by the concentration of $\hat{C}_r$ around its mean.
Let $E_{\beta}=\{w\in \mathbb{C}^n, |w^1_j|^2\leq  2\beta \log(n) \text{ and }  |w^2_j|^2\leq 2\beta \log(n),j=1\dots n\}.$
Let $(w^1_i,w^2_i), i=1\dots r $, be $2r$ independent iid $\mathcal{C}\mathcal{N}(0,I_n)$.
Define $(\tilde{w}^1_i,\tilde{w}^2_i)=(w^1_i,w^2_i)$ if $(w^1_i,w^2_i)\in E_{\beta}$ and $(\tilde{w}^1_i,\tilde{w^2}_i)=(0,0)$ elsewhere.\\
Define $\tilde{C}_r=\frac{1}{r}\sum_{i=1}^r \tilde{A}_i$, where $$\tilde{A_i}=Diag(\tilde{w}^1_i)FDiag(\tilde{y}_i)F^*Diag(\tilde{w}^{1,*}_i)-Diag(\tilde{w}^2_i)FDiag(\tilde{y}_i)F^*Diag(\tilde{w}^{2,*}_i),$$
and let $\tilde{C}=\mathbb{E}\tilde{C}_r$.
By the triangular inequality we have:
\begin{equation}
\left|\left| \hat{C}_r-C\right|\right|\leq \left|\left| \hat{C}_r-\tilde{C}_r\right|\right|+\left|\left|\tilde{C}_r-\tilde{C}\right|\right|+\left|\left|\tilde{C}-C\right|\right| 
\label{eq:triangle}
\end{equation}
\noindent \textbf{Bounding $\left|\left| \hat{C}_r-\tilde{C}_r\right|\right|$:  }\\
Note that $\left|\left| \hat{C}_r-\tilde{C}_r\right|\right|=0$, if for all $i=1\dots r$, $(w^1_i,w^2_i)\in E_{\beta}$.
Let us get a  bound on the probability of that event.
Note that: $\mathbb{P}(|w_i|^2>2\beta \log(n))\leq 2n^{-\beta}$.
To avoid cumbersome notations when we use index $i$, $w_i$ refers to a modulation in $\mathbb{C}^n$, $i=1\dots r$, and when we use index $j$ $w_j$ refers to the $j-th$ component of $w\in \mathbb{C}^n, j=1\dots n$.

\begin{eqnarray*}
\mathbb{P}\left\{ \exists i \in \{1\dots r\} \quad \text{ such that }  (w^1_i,w^2_i) \notin E_{\beta}\right\}&\leq& r \mathbb{P}\left\{ (w^1,w^2) \notin E_{\beta}\right\}\\
&=&r \mathbb{P}\left\{\exists j \text{ such that } |w^1_j|^2> 2\beta \log(n) \text{ Or }  |w^1_j|^2> 2\beta \log(n)\right\}\\
&\leq&  \frac{4rn}{n^{\beta}}.
\end{eqnarray*}
It follows that:
\begin{equation}
\left|\left| \hat{C}_r-\tilde{C}_r\right|\right|=0 \text{ with probability at least } 1-\frac{4r}{n^{\beta-1}}.
\label{eq:coincide}
\end{equation}

\noindent \textbf{Bounding $\left|\left|\tilde{C}_r-\tilde{C}\right|\right|$:}\\

\noindent Let $$\tilde{X}_i = \tilde{A}_i-\mathbb{E}(\tilde{A}_i),$$
$\mathbb{E}(\tilde{X}_i)=0$. 
Note that $$\left|\left|\tilde{C}_r-\tilde{C}\right|\right|= \frac{1}{r}\left|\left| \sum_{i=1}^r\tilde{X}_i\right|\right|$$
Let us get a bound on $||\tilde{X}_i||$.
Note that $||\mathbb{E}(\tilde{A_i})||\leq \lambda$.
To simplify the notation we will omit in the following the indices.
\begin{eqnarray*}
\left|\left|  \tilde{A}\right|\right| &=& \sup_{x,||x||=1} \scalT{\tilde{y}}{|F Diag(\tilde{w}^1) x |^2-|FDiag(\tilde{w}^2)x|^2} 
\end{eqnarray*}
Recall that $v=F^* Diag(\tilde{w}^*)x$.
Note that: $|FDiag(\tilde{w})x|^2= diag\left\{  F^* Diag(\tilde{w}^*)xx^*Diag(\tilde{w})F \right\}= diag(vv^*)$.\\
By holder inequality we have:
\begin{eqnarray*}
 \scalT{\tilde{y}}{|F Diag(\tilde{w}^1) x |^2-|FDiag(\tilde{w}^2)x|^2} &=& \scalT{\tilde{y}}{diag( v^{1}v^{1,*})-diag( v^{2}v^{2,*})}\\
 &\leq& ||\tilde{y}||_{\infty} \left|\left|diag( v^{1}v^{1,*})-diag( v^{2}v^{2,*})\right|\right|_{\ell_1}. 
 \end{eqnarray*}
 Since $\tilde{y}$ is binary  $||\tilde{y}||_{\infty} =1$.
 By the triangular inequality:$$ \left|\left|diag( v^{1}v^{1,*})-diag( v^{2}v^{2,*})\right|\right|_{\ell_1} \leq \left|\left|diag( v^{1}v^{1,*})\right|\right|_{\ell_1} +\left|\left|diag( v^{2}v^{2,*})\right|\right|_{\ell_1}.$$
 \begin{eqnarray*}
\left|\left|diag( vv^{*})\right|\right|_{\ell_1}=\left|\left||FDiag(w)x|^2\right|\right|_{\ell_1} &=& ||diag\left\{  F^* Diag(\tilde{w}^*)xx^*Diag(\tilde{w})F \right\} ||_{\ell_1}\\
&=&Tr(F^* Diag(\tilde{w}^*)xx^*Diag(\tilde{w})F) \\
&=& Tr(Diag(\tilde{w})FF^*Diag(\tilde{w}^*)xx^*)\\
&=& Tr(Diag(|\tilde{w}|^2)xx^*)\\
&\leq & \left|\left| Diag(|\tilde{w}|^2)\right|\right| ||x||^2.
\end{eqnarray*}
\noindent We are now left with Bounding: 
\begin{eqnarray}
\sup_{x,||x||=1}  \left(\left|\left| Diag(|\tilde{w}^1|^2)\right|\right|+\left|\left| Diag(|\tilde{w}^2|^2)\right|\right| \right) ||x||^2 &=&  \left(\left|\left| Diag(|\tilde{w}^1|^2)\right|\right|+\left|\left| Diag(|\tilde{w}^2|^2)\right|\right| \right)\nonumber\\
&=&\max_{j=1\dots n }|\tilde{w}^1_j|^2 + \max_{j=1\dots n} |\tilde{w}^2_j|^2  
\label{eq:max}
\end{eqnarray}
By definition of $(\tilde{w}^1,\tilde{w^2})$ we conclude that : 
\begin{equation}
||\tilde{A}||\leq 4\beta \log(n).
\end{equation}
It follows that $$||\tilde{X}_i||\leq 4\beta\log(n)+\lambda\leq 5\beta \log(n) :=\Delta.$$
\begin{theorem}[Hoeffding Matrix Inequality \cite{tropp2012user}]
Let $X_{i},i=1\dots r $ be a sequence of independent random $n\times n$ self adjoint matrices.
Assume that each random matrix obeys:
$$\mathbb{E}(X_i)=0\quad \text{ and } ||X_i||\leq \Delta \text{ almost surely.}$$
Then for all $t\geq0$,
$$\mathbb{P}\left\{\frac{1}{r}\left|\left|\sum_{i=1}^rX_i \right|\right|\geq t\right\}\leq 2n\exp\left(-\frac{rt^2}{8\Delta^2}\right).$$ 
In other words:
$$\text{For } r \geq \frac{t^2}{\epsilon^2} \Delta^2\log(n), \quad \frac{1}{r}\left|\left|\sum_{i=1}^rX_i \right|\right|\leq \epsilon \text{ with probability at least } 1-n^{-t^2} . $$
\end{theorem}
\noindent We are now  ready to apply the Hoeffding Matrix inequality:
$$\text{For } r \geq c \frac{t^2}{\epsilon^2} \beta^2\log^3n, \quad \frac{1}{r}\left|\left|\sum_{i=1}^r \tilde{X}_i \right|\right|\leq \epsilon \text{ with probability at least } 1-n^{-t^2} . $$
It follows that:
\begin{equation}
\text{For }\left|\left|\tilde{C}_r-\tilde{C}\right|\right| \leq ct\beta \sqrt{\frac{\log^3n}{r} } \text{ with probability at least } 1-n^{-t^2} . 
\label{eq:Hoeffding}
\end{equation}
\noindent \textbf{Bounding $\left|\left|\tilde{C}-{C}\right|\right|$:}\\
By Jensen inequality followed by Cauchy Sharwz inequality we have:\\
\begin{eqnarray*}
\left| \left  |\tilde{C}-{C}\right|\right|&=&\left| \left| \mathbb{E}\left(\ind_{w^1,w^2 \notin E_{\beta}} A\right)\right|\right|\\
&\leq&\mathbb{E}  \ind_{w^1,w^2 \notin E_{\beta}}\left| \left|  A\right|\right|\\
&\leq&\sqrt{\mathbb{E}( \ind_{w^1,w^2 \notin E_{\beta}})}\sqrt{E(\left| \left|  A\right|\right|^2)}\\
&\leq&  \sqrt{\mathbb{P}(E^c_{\beta})}\sqrt{\mathbb{E}(\max_{j=1\dots n} |w^1_j|^2+ \max_{j=1\dots n} |w^2_j|^2)^2}.
\end{eqnarray*}
The last inequality follows from equation \eqref{eq:max}.
\begin{equation}
||A||^2\leq (\max_{j=1\dots n }|w^1_j|^2)^2 + (\max_{j=1\dots n} |w^2_j|^2)^2+2\max_{j=1\dots n }|w^1_j|^2  \max_{j=1\dots n} |w^2_j|^2.
\end{equation}
Let $Z=\max_{j=1\dots n}E_j, \quad E_j \sim Exp(1)$ iid, therefore:
\begin{equation}
\mathbb{E}\left(||A||^2\right)\leq 2(E(Z^2)+(E(Z))^2)= 2 (Var(Z))+2(E(Z))^2).
\end{equation}
\begin{lemma}[Maximum of Exponential \cite{bousquet}]
Let $Z=\max_{j=1\dots n}E_j, \quad E_j \sim Exp(1)$ iid, therefore:
$Var(Z)\leq 2 \quad \mathbb{E}(Z)=\sum_{i=1}^n \frac{1}{i} \leq \log(n).$
\end{lemma}
\noindent For sufficiently large $n$, there exists a constant $c$ such that:
\begin{equation}
\mathbb{E}\left(||A||^2\right)\leq 2(E(Z^2)+(E(Z))^2)= 2 (2+2(\log^2(n)))=4(1+\log^2(n))\leq c^2\log^2(n).
\end{equation}
Note that $\mathbb{P}(E^c_{\beta})\leq \frac{4n}{n^{\beta}}$.
Therefore:\\
\begin{equation}
\left| \left  |\tilde{C}-{C}\right|\right| \leq \frac{{2}c}{n^{(\beta-1)/2}}\log(n).
\label{eq:Exp}
\end{equation}
\textbf{Putting all together:}\\
Putting together equations \eqref{eq:triangle},\eqref{eq:coincide},\eqref{eq:Hoeffding} and  \eqref{eq:Exp} we have finally with probability at least $1-n^{-t^2}-\frac{4r}{n^{\beta-1}}$:
\begin{equation}
||\hat{C}_r-C|| \leq c\beta t \sqrt{\frac{\log^3n}{r}}+ \frac{{2}c}{n^{(\beta-1)/2}}\log(n).
\end{equation}
Setting $t=\sqrt{2},\beta=4$  we get with probability $1-O(n^{-2})$,
\begin{equation}
||\hat{C}_r-C|| \leq 4\sqrt{2}c \sqrt{\frac{\log^3n}{r}}+ \frac{{2}c}{n^{3/2}}\log(n).
\end{equation}
It follows that there exists a numeric constant $c$ such that: 
\begin{equation}
\text{For } r\geq c \frac{\log^3n}{\epsilon^2}\quad ||\hat{C}_r-C|| \leq\epsilon, \text{ with probability at least  } 1-O(n^{-2}).
\end{equation}
Finally by equation \eqref{eq:compa} we conclude that:
\begin{equation}
\text{For } r\geq c \frac{\log^3n}{\epsilon^2}\quad \frac{1}{2}||\hat{x}_r\hat{x}^*_{r}-x_0x_0^*||^2_{F} \leq\frac{\epsilon}{\lambda}, \text{ with probability at least  } 1-O(n^{-2}).
\end{equation}
In other words for another numeric constant $c$ we have:
\begin{equation}
\text{For } r\geq c \frac{\log^3n}{\epsilon^2\lambda^2}\quad ||\hat{x}_r\hat{x}^*_{r}-x_0x_0^*||^2_{F} \leq \epsilon, \text{ with probability at least  } 1-O(n^{-2}).
\end{equation}

\newcommand{\etalchar}[1]{$^{#1}$}

\end{document}